\documentclass[12pt,reqno,a4paper]{amsart}
\usepackage[british]{babel}
\usepackage{amssymb,amsthm}
\usepackage[utf8]{inputenc}
\usepackage[margin=2.8cm]{geometry}
\usepackage[small]{eulervm}
\usepackage{tgpagella}
\usepackage{bbold} 
\usepackage{array}
\usepackage{enumitem}
\usepackage[vcentermath]{youngtab}
\Yboxdim{4pt}
\usepackage[unicode]{hyperref}
\hypersetup{%
  pdftitle   = {Spencer cohomology and eleven-dimensional supergravity},
  pdfkeywords = {Spencer cohomology, Lie superalgebras, Poincaré, supergravity, maximal supersymmetry},
  pdfauthor  = {José Figueroa-O'Farrill, Andrea Santi},
  pdfcreator = {\LaTeX\ with package \flqq hyperref\frqq}
}
\theoremstyle{plain}
\newtheorem{lemma}{Lemma}
\newtheorem{proposition}[lemma]{Proposition}
\newtheorem{theorem}[lemma]{Theorem}

\theoremstyle{definition}
\newtheorem*{remark}{Remark}
\newcommand{\ten}{\natural}
\newcommand{\Hom}{\mathrm{Hom}}
\newcommand{\End}{\mathrm{End}}
\newcommand{\AdS}{\mathrm{AdS}}
\newcommand{\Gr}{\mathrm{Gr}}
\newcommand{\GL}{\mathrm{GL}}
\newcommand{\SO}{\mathrm{SO}}
\newcommand{\CSO}{\mathrm{CSO}}
\newcommand{\Spin}{\mathrm{Spin}}
\newcommand{\CSpin}{\mathrm{CSpin}}
\renewcommand{\Im}{\mathrm{Im}}
\newcommand{\Cl}{C\ell}
\newcommand{\ad}{\mathrm{ad}}
\newcommand{\dvol}{\mathrm{dvol}}
\newcommand{\stab}{\mathfrak{stab}}
\newcommand{\fosp}{\mathfrak{osp}}
\newcommand{\fso}{\mathfrak{so}}
\newcommand{\fm}{\mathfrak{m}}
\newcommand{\fg}{\mathfrak{g}}
\newcommand{\fh}{\mathfrak{h}}
\newcommand{\1}{\mathbb{1}} 
\newcommand{\PP}{\mathbb{P}}
\newcommand{\RR}{\mathbb{R}}
\newcommand{\ZZ}{\mathbb{Z}}
\newcommand{\NN}{\mathbb{N}}
\newcommand{\CC}{\mathbb{C}}
\newcommand{\be}{\boldsymbol{e}}
\newcommand{\sbar}{{\overline s}}

\allowdisplaybreaks
\begin{document}

\title{Spencer cohomology and eleven-dimensional supergravity}
\author{José Figueroa-O'Farrill}
\author{Andrea Santi}
\address{Maxwell Institute and School of Mathematics, The University
  of Edinburgh, James Clerk Maxwell Building, Peter Guthrie Tait Road,
  Edinburgh EH9 3FD, United Kingdom}
\thanks{EMPG-15-22}
\begin{abstract}
  We recover the classification of the maximally supersymmetric
  bosonic backgrounds of eleven-di\-men\-sional supergravity by Lie
  algebraic means. We classify all filtered deformations of the
  $\mathbb Z$-graded subalgebras
  $\fh=\fh_{-2}\oplus\fh_{-1}\oplus\fh_{0}$ of the Poincaré
  superalgebra $\fg=\fg_{-2}\oplus\fg_{-1}\oplus\fg_{0}=V\oplus
  S\oplus \fso(V)$ which differ only in zero degree, that is
  $\fh_0\subset\fg_0$ and $\fh_j=\fg_j$ for $j<0$. Aside from the
  Poincaré superalgebra itself and its $\mathbb Z$-graded subalgebras,
  there are only three other Lie superalgebras, which are the symmetry
  superalgebras of the non-flat maximally supersymmetric backgrounds.
  In passing we identify the gravitino variation with (a component of)
  a Spencer cocycle.
\end{abstract}
\maketitle
\tableofcontents

\section{Introduction}
\label{sec:introduction}

The work described in this paper is an attempt at breaking new ground
in the classification problem of supersymmetric backgrounds of
eleven-di\-men\-sional supergravity.  This problem as such has been
pursued on and off for the last 15 years; although its roots date back
to the 1980s and the classification results for Freund--Rubin-like
backgrounds (see, e.g., the review \cite{DNP}) in the context of
Kaluza--Klein supergravity.

A convenient organising principle in the classification of
supersymmetric supergravity backgrounds is the fraction $\nu$ of
supersymmetry preserved by the background, which is ``categorified''
as the dimension $N = 32 \nu$ of the odd subspace of its Killing
superalgebra.  At present there exist a classification for $N=32$
\cite{FOPMax}, non-existence results for $N=31$
\cite{NoMPreons,FigGadPreons} and $N=30$ \cite{Gran:2010tj}, a
structure result for $N=1$ \cite{GauPak,GauGutPak}, and a huge zoo of
solutions for other values of $N$, but no claim of classification.  No
solutions are known for $30>N>26$, but there is a pp-wave background
with $N=26$ \cite{Michelson26}.  This ``supersymmetry gap'' is
reminiscent of the gap phenomenon in geometric structures (see, e.g.,
\cite{2013arXiv1303.1307K,deMedeiros:2014hla}) and, indeed, part of
the motivation to explore the approach presented in this paper was to
understand the nature of this gap.

The consensus seems to be that, at present, the classification of all
supersymmetric backgrounds is inaccessible, whereas that of highly
supersymmetric backgrounds seems tantalisingly in reach.  In
particular, backgrounds with $N>16$ are now known to be locally
homogeneous \cite{HomogThm} and this brings to bear the techniques of
homogeneous geometry to classify certain kinds of backgrounds; e.g.,
symmetric \cite{FigueroaO'Farrill:2011fj,Hustler:2015lca} or
homogeneous under a given Lie group
\cite{FOUM2,2014arXiv1409.2664L}, at least when the group is
semisimple.

This paper is a first step in a Lie algebraic approach at the
classification problem.  The proposal, to be made more precise in a
forthcoming paper, is to take the Killing superalgebra as the
organising principle.  As we will show in that forthcoming paper, the 
Killing superalgebra of a supersymmetric eleven-di\-men\-sional
supergravity background (and also, indeed, of other supergravity
theories) is a filtered deformation of a subalgebra of the relevant
Poincaré superalgebra.  The classification problem of filtered
deformations of Lie superalgebras seems tractable via cohomological
techniques \cite{MR1688484,MR2141501} which extend the use of Spencer
cohomology in the theory of $G$-structures or, more generally, Tanaka
structures.

Therefore in this paper we will present a Lie algebraic derivation of (the
symmetry superalgebras of) the maximally supersymmetric bosonic
backgrounds of eleven-di\-men\-sional supergravity by purely
representation theoretic means.  In so doing we will actually
``rediscover'' eleven-di\-men\-sional supergravity from a cohomological
calculation.

Our point of departure will be the Poincaré superalgebra
$\fg = \fg_{-2} \oplus \fg_{-1} \oplus \fg_{0} = V \oplus S \oplus
\fso(V)$ or, more precisely, its supertranslation ideal $\fm =
\fm_{-2}\oplus\fm_{-1} = V\oplus S$. At first, it might seem
overoptimistic to expect that such a derivation is possible. How does
the supertranslation ideal (or even the Poincaré superalgebra) know
about the maximally supersymmetric supergravity backgrounds?  We can
give at least two heuristic answers to this question.

The physicist's answer is that, in a sense, this has always been
possible, albeit via a rather circuitous route.  That route starts by
searching for massless irreducible unitary representations of the
Poincaré superalgebra.  Following Nahm \cite{Nahm}, we would find the
``supergravity multiplet'': the unitary irreducible representation
induced from the (reducible) representation of the ``little group''
$\Spin(9)$ isomorphic to
\begin{equation*}
  \odot^2_0W \oplus \Lambda^3 W \oplus (W \otimes \Sigma)_0~,
\end{equation*}
where $W$ and $\Sigma$ are, respectively the real 9-di\-men\-sional vector
and 16-di\-men\-sional spinor representation of $\Spin(9)$, $\odot^2_0$
means symmetric traceless and the subscript $0$ on the last term means
the kernel of the Clifford action $W \otimes \Sigma \to \Sigma$ or,
equivalently, ``gamma traceless''.  (More generally, we use the
notation $\odot^n$ to mean the $n$-th symmetric tensor power.)  In this
data, a physicist would recognise at once the physical degrees of
freedom corresponding to a Lorentzian metric $g$, a 3-form potential
$A$ and a gravitino $\Psi$ and would set to construct a supergravity
theory with that field content.  It turns out that there is a unique
such supergravity theory, which was constructed by Cremmer, Julia and
Scherk in \cite{CJS}.  The action (with $\Psi=0$) is given by the sum
\begin{equation*}
  I = I_{\textrm{EH}} + I_{\textrm{M}} + I_{\textrm{CS}} = \tfrac12
  \int R \dvol + \tfrac14 \int F \wedge \star F + \tfrac1{12} \int F
  \wedge F \wedge A~,
\end{equation*}
where $F =dA$, of an Einstein--Hilbert, Maxwell and Chern--Simons
actions. The full action (including the terms depending on the
gravitino $\Psi$) is invariant under local supersymmetry. The
transformation of the gravitino under local supersymmetry defines a
connection $D$ on the spinor bundle, which encodes most of the
geometric data of the supergravity theory. For all vector fields $X$
and spinor fields $\varepsilon$, the connection $D$ is defined
by
\begin{equation}
  \label{eq:D}
  D_X \varepsilon = \nabla_X \varepsilon + \tfrac16 \iota_X F \cdot
  \varepsilon + \tfrac1{12} X^\flat \wedge F \cdot \varepsilon~,
\end{equation}
with $X^\flat$ the dual one-form to $X$ and $\cdot$ denoting the
Clifford action. 

A maximally supersymmetric bosonic background is one where $\Psi=0$
and $D$ is flat (one checks that $D$-flatness actually implies the
field equations). The $D$-flatness equations can be solved and one
finds, as was done in \cite{FOPMax}, that besides Minkowski spacetime
(with $F=0$) there are three further families of backgrounds: two
one-parameter families of Freund--Rubin backgrounds --- the original
background $\AdS_4 \times S^7$ found by Freund and Rubin in
\cite{FreundRubin} and $\AdS_7 \times S^4$, found by Pilch, Townsend
and van~Nieuwenhuizen in \cite{AdS7S4} --- and a symmetric pp-wave
found by Kowalski-Glikman in \cite{KG} and interpreted in
\cite{ShortLimits} as the Penrose limit of the Freund--Rubin
backgrounds. The calculation of the symmetry superalgebra of these
backgrounds is then straightforward and we arrive at the Poincaré
superalgebra itself for the Minkowski background, orthosymplectic Lie
superalgebras $\fosp(8|4)$ for $\AdS_4 \times S^7$ and $\fosp(2,6|4)$
for $\AdS_7 \times S^4$, and a contraction thereof for the
Kowalski-Glikman wave (see
\cite{JMFKilling,FOPflux,HatKamiSaka}). Although all of these
backgrounds are maximally supersymmetric, it is the Minkowski
background which has the largest symmetry: the Poincaré superalgebra
has dimension $(66|32)$, whereas the symmetry superalgebras of the
other maximally supersymmetric bosonic backgrounds have dimension
$(38|32)$.

That would be the physicist's answer, but there is also a geometer's
answer to the question of how the supertranslation ideal knows about
(the symmetry superalgebras of) the maximally supersymmetric
backgrounds, stemming from the integrability problem for geometric
structures.  The point of departure in this story is the fact that
eleven-di\-men\-sional supergravity also admits, besides the traditional
``component'' formulation, a geometric presentation in terms
of supermanifolds.  This usually amounts to giving a reduction to
$\Spin(V)$ of the linear frame bundle of a supermanifold $M$ of
dimension $(11|32)$ or, in other words, a $G$-structure $\pi:P\to M$
where the structure group $G=\Spin(V)$ acts on the vector space direct
sum $V\oplus S$ of its vector and spinor representations (see
\cite{Brink:1980az,Cremmer:1980ru}). The geometric structure under
consideration is \emph{not arbitrary} but it satisfies some
constraints, expressed in terms of appropriate nondegeneracy
conditions on the intrinsic torsion of $\pi:P\to M$ (see, e.g.,
\cite{MR1079796} for a geometric motivation of the constraints). The
constraints put the theory ``on-shell'', in the sense that every
$G$-structure as above gives rise to a solution of the field equations
(see \cite{Howe:1997he}).

It appears therefore that the ``vacuum solution'' given by the super
Minkowski spacetime is described by a geometric structure which is not
integrable or flat, at least in the sense of $G$-structures. In
particular (the super-analogues of) the classical Spencer cohomology
groups and their associated intrinsic curvatures considered in
\cite{MR0203626} are not applicable to the study of the deformations
of these structures and, ultimately, to the quest for supergravity
backgrounds.

In \cite{MR3056953,MR2798219} a description of eleven-di\-men\-sional
supergravity based on the notion of a super Poincaré structure is
proposed. This is an odd distribution $\mathcal D\subset TM$ on a
supermanifold $M$ of dimension $(11|32)$ which is of rank $(0|32)$ and
with Levi form
\begin{equation}
  \mathcal L : \mathcal D \otimes \mathcal D \to TM/\mathcal
  D\;,\qquad\mathcal L(X,Y)=[X,Y]\mod\mathcal D\;,
\end{equation}
locally identifiable with the bracket $S\otimes S\to V$ of the
supertranslation algebra $\fm$. Note that $\mathcal D$ is a maximally
nonintegrable distribution and it is of depth $d=2$, in the sense that
$\Gamma(\mathcal D)+[\Gamma(\mathcal D),\Gamma(\mathcal
D)]=\mathfrak{X}(M)$.  These structures can be studied with (the
analogues for supermanifolds of) the standard techniques of the theory
of Tanaka structures, a powerful generalisation of $G$-structures
found by Tanaka in \cite{MR0266258,MR533089} to deal with geometries
supported over non-integrable distributions.

Let us briefly recall the main points of Tanaka's approach. It builds
on the observation that a distribution $\mathcal D$ on a manifold $M$
determines a filtration
\begin{equation*}
  T_xM = \mathcal D_{-d}(x) \supset \mathcal D_{-d+1}(x) \supset
  \cdots \supset \mathcal D_{-2}(x) \supset\mathcal D_{-1}(x) =
  \mathcal D_x
\end{equation*}
of each tangent space $T_xM$, $\mathcal D_{-i}(x)$ being the subspace
of $T_xM$ given by the values of the vector fields in
$\Gamma(\mathcal D)_{-i}=\Gamma(\mathcal D)_{-i+1}+[\Gamma(\mathcal
D)_{-1},\Gamma(\mathcal D)_{-i+1}]$
and $\Gamma(\mathcal D)_{-1}=\Gamma(\mathcal D)$. He then noticed that
the ``symbol space''
\begin{equation*}
  \fm(x) = \operatorname{gr}(T_xM) = \fm_{-d}(x) \oplus \fm_{-d+1} (x)
  \oplus \cdots \oplus \fm_{-2}(x) \oplus \fm_{-1}(x)
\end{equation*}
inherits the structure of a $\mathbb Z$-graded Lie algebra by the
commutators of vector fields and assumed, as a regularity condition,
that all $\fm(x)$ are isomorphic to a fixed $\mathbb Z$-graded Lie
algebra $\fm=\fm_{-d}\oplus\cdots\oplus\fm_{-1}$ which is generated by
$\fm_{-1}$.  We call such $\ZZ$-graded Lie algebras
\emph{fundamental}.

To any fundamental Lie algebra $\fm$ one can associate a unique
\emph{maximal transitive prolongation} in positive degrees
\begin{equation*}
  \fg^\infty=\bigoplus_{p\in \mathbb Z} \fg^\infty_p\;.
\end{equation*}
This is a (possibly infinite-di\-men\-sional) $\mathbb{Z}$-graded Lie algebra
which satisfies:
\begin{enumerate}
\item $\fg^\infty_p$ is finite-di\-men\-sional for every $p\in\mathbb Z$;
\item $\fg^\infty_p=\fm_p$ for every $-d\leq p\leq -1$ and $\fg^\infty_p=0$ for every $p<-d$;
\item for all $p\geq 0$, if $x\in\fg^\infty_p$ is an element such that
  $[x,\fg^\infty_{-1}]=0$, then $x=0$ (this property is called \emph{transitivity});
\item $\fg^\infty$ is \emph{maximal} with these properties.
\end{enumerate}
Finally he introduced the concept of a Tanaka structure, a
$G_0$-reduction $\pi:P\rightarrow M$ of an appropriate
$G_0^\infty$-principal bundle, $Lie(G^\infty_0)=\fg^\infty_0$,
consisting of linear frames defined just on the subspaces
$\mathcal D_x$ of the $T_xM$'s (in particular the usual $G$-structures
are the Tanaka structures of depth $d=1$), and also the analogs of the
Spencer cohomology groups and their associated intrinsic curvatures
for $\mathbb Z$-graded Lie algebras of depth $d>1$. In this context
the integrable model, which realises the maximum dimension of the
algebra of symmetries, is the nilpotent and simply connected Lie group
with Lie algebra $\fm$.
 
In the relevant case of supermanifolds and eleven-di\-men\-sional
supergravity, the symbol is just the supertranslation algebra $\fm$,
the integrable model is the super Minkowski spacetime and a Tanaka
structure on a supermanifold with symbol $\fm$ and structure group
$G_0=\Spin(V)\subset G^\infty_0=\CSpin(V)$ is the same as a
supergravity background.

This paper considers the deformations of the super Minkowski spacetime
from Tanaka's perspective and recovers the classification of maximally
supersymmetric bosonic backgrounds of eleven-di\-men\-sional supergravity
by Lie algebraic means. Our starting point is the supertranslation
algebra $\fm=\fm_{-2}\oplus\fm_{-1}=V\oplus S$ and the nontrivial result 
\cite{MR3218266,MR3255456} that its maximal transitive prolongation is
the extension
\begin{equation*}
  \fg^\infty = \fg_{-2}^\infty \oplus \fg_{-1}^\infty \oplus
  \fg_0^\infty = V \oplus S \oplus \left(\fso(V) \oplus \mathbb R E\right)
\end{equation*}
of the Poincaré superalgebra
$\fg=\fg_{-2}\oplus\fg_{-1}\oplus\fg_{0}=V\oplus S\oplus \fso(V)$ by 
the grading element
\begin{equation*}
  E\in\fg^\infty_0 \qquad\text{where}\qquad
  \left.\ad(E)\right|_{\fg^\infty_j}=j \, Id_{\fg^\infty_j}~.
\end{equation*}

More precisely we will show that the symmetry superalgebras of the
maximally supersymmetric bosonic backgrounds correspond exactly to the
filtered deformations of the $\mathbb Z$-graded subalgebras
$\fh=\fh_{-2}\oplus\fh_{-1}\oplus\fh_{0}$ of the Poincaré
superalgebra which differ only in zero degree, that is
$\fh_0\subset\fg_0$ and $\fh_j=\fg_j$ for $j<0$. We will make evident
that a gap phenomenon arises, the dimension of the symmetry
superalgebra dropping when considering non-integrable geometries, and,
in doing so, we also recover the connection \eqref{eq:D} by cohomological
methods.  In other words, we rediscover the basic geometric object of
the supergravity theory by a cohomological calculation.

We remark that similar gaps and upper bounds on the submaximal
dimension for (non-super) geometric Tanaka structures were recently
derived in \cite{2013arXiv1303.1307K} using Kostant's version
\cite{MR0142696} of Borel--Bott--Weil theory for semisimple Lie
algebras, whereas in our case we require different cohomological
techniques, developed for general $\ZZ$-graded Lie superalgebras by
Cheng and Kac in \cite{MR1688484,MR2141501}.

The paper is organised as follows.  In Section \ref{sec:poinc-super}
we introduce the problem by defining the notion of a filtered
deformation of $\mathbb Z$-graded subalgebras $\fh$ of the Poincaré
superalgebra $\fg$ differing only in degree $0$.  We observe that
infinitesimal filtered deformations can be interpreted in terms of
Spencer cohomology.  In Section~\ref{sec:spencer-complex} we introduce
the Spencer differential complex $C^{\bullet,\bullet}(\fm,\fg)$ and
prove that $H^{p,2}(\fm,\fg)=0$ for all even $p\geq 4$.  The main
result of Section~\ref{sec:infin-deform-g} is
Proposition~\ref{prop:beta}, giving an explicit isomorphism of
$\fso(V)$-modules between the group $H^{2,2}(\fm,\fg)$ and
$\Lambda^4 V$.  With only a modicum of hyperbole, we explain that we
may interpret this result as a cohomological derivation of
eleven-di\-men\-sional supergravity.  In
Section~\ref{sec:infin-deform-h} we consider the subalgebras $\fh$,
determine the corresponding Spencer groups $H^{p,2}(\fm,\fh)$ for all
even $p\geq 2$ and prove Theorem~\ref{thm:first}, which states that
infinitesimal filtered deformations of $\fh$ are classified by
$\fh_0$-invariant elements in $H^{2,2}(\fm,\fh)$.  In
Section~\ref{sec:integ-deform} we determine the $\fh_0$-invariant
elements in $H^{2,2}(\fm,\fh)$ and integrate the corresponding
infinitesimal deformations.  The classification of infinitesimal
deformations is contained in Proposition~\ref{thm:second} and their
integrability is proved in Section~\ref{sec:first-order-integr-deform}
and Section~\ref{sec:all-order-integr}.  Our results are summarised in
Theorem~\ref{thm:final}. The paper ends with some discussions in
Section~\ref{sec:disc-concl}.  Finally,
Appendix~\ref{sec:clifford-conventions} and
Appendix~\ref{sec:some-representations} set our conventions and basic
results on Clifford algebras, spinors and representations of
$\fso(V)$.

\section{The deformation complex}
\label{sec:deformation-cohomology}

In this section we give the basic definitions and then prove the first
results on the Spencer cohomology of the Poincaré superalgebra.

\subsection{The Poincaré superalgebra}
\label{sec:poinc-super}

Let $V$ denote a real eleven-di\-men\-sional vector space with a
Lorentzian inner product $\eta$ of signature $(1,10)$; that is, $\eta$
is ``mostly minus''.  The corresponding Clifford algebra $\Cl(V) \cong
\Cl(1,10) \cong \End(S_+) \oplus \End(S_-)$ where $S_\pm$ are irreducible
Clifford modules, real and of dimension $32$.  They are distinguished
by the action of the volume element $\Gamma_{11}\in\Cl(V)$, but are
isomorphic as $\Spin(V)$ representations.  We will work with $S=S_-$
in what follows, that is we assume $\Gamma_{11}\cdot s=-s$ for all
$s\in S$.

On $S$ there is a symplectic structure $\left<-,-\right>$ satisfying
\begin{equation}
  \label{eq:symplectic}
  \left<v\cdot s_1, s_2\right> = - \left<s_1, v \cdot s_2\right>~,
\end{equation}
for all $s_1,s_2 \in S$ and $v \in V$, where $\cdot$ refers to the
Clifford action.  In particular, $\left<-,-\right>$ is $\Spin(V)$-invariant,
making $S$ into a real symplectic representation of $\Spin(V)$.
Taking adjoint with respect to the symplectic structure defines an
anti-involution $\sigma$ on $\Cl(V)$ which, by \eqref{eq:symplectic},
is characterised by $\left.\sigma\right|_V = - \operatorname{Id}_V$.

Let $\fso(V)$ denote the Lie algebra of $\Spin(V)$.  The
(eleven-di\-men\-sional) Poincaré superalgebra is the $\mathbb Z$-graded Lie
superalgebra
\begin{equation*}
  \fg = \fg_{-2} \oplus \fg_{-1} \oplus \fg_{0}~,
\end{equation*}
where $\fg_0 = \fso(V)$, $\fg_{-1} = S$ and $\fg_{-2} = V$. The
$\ZZ$-grading is \emph{compatible} with the parity, in the sense that
$\fg_{\bar 0}=\fg_{-2}\oplus\fg_{0}$ and $\fg_{\bar 1}=\fg_{-1}$,
and it allows only the following brackets:
\begin{itemize}
\item $[-,-]: \fg_0 \times \fg_i \to \fg_i$, which consists of the
  adjoint action of $\fso(V)$ on itself and its natural actions on $V$
  and $S$;
\item $[-,-]: \fg_{-1} \times \fg_{-1} \to \fg_{-2}$, which is the
  construction of the Dirac current of a spinor:
  \begin{equation}
    \label{eq:DiracCurrent}
    \eta(v,[s_1,s_2]) = \left<v\cdot s_1, s_2\right>~,
  \end{equation}
  for all $s_1,s_2 \in S$ and $v \in V$.  Notice that from
  \eqref{eq:symplectic}, it follows that $[s_1,s_2] = [s_2,s_1]$ and
  hence that it is determined by its restriction $[s,s]$ to the
  diagonal.  It is a fact that for all $s \in S$, the vector $v=[s,s]
  \in V$, satisfies $\eta(v,v) \geq 0$, i.e. it is either null or timelike.
\end{itemize}
Note that the even Lie subalgebra $\fg_{\bar 0}=\fso(V) \oplus V$ is
the Poincaré algebra.  We will let $\fm= \fm_{-2} \oplus \fm_{-1}$,
$\fm_{-2}=\fg_{-2}=V$, $\fm_{-1}=\fg_{-1}=S$ denote the (2-step
nilpotent) supertranslation ideal.  As it is generated by $\fm_{-1}$,
$\fm$ is a fundamental $\mathbb Z$-graded Lie superalgebra.

We consider $\mathbb Z$-graded subalgebras
$\fh=\fh_{-2}\oplus\fh_{-1}\oplus\fh_0$ of the Poincaré superalgebra
which differ only in zero degree, that is, $\fh \subset \fg$ with
$\fh_0 \subset \fg_0$ and $\fh_j = \fg_j$ for $j<0$ and we seek
filtered deformations of $\fh$. These are the Lie superalgebras $F$
with an associated compatible filtration
$F_\bullet=\cdots \supset F_{-2} \supset F_{-1} \supset F_0\supset
\cdots$
such that the corresponding $\mathbb Z$-graded Lie superalgebra agrees
with $\fh$ (see, e.g., \cite{MR1688484,MR2141501}).  Any such
filtration $F_\bullet$ is isomorphic \emph{as a vector space} to the
canonical filtration of $\fh$ given by $F_{i} = \fh$ for all $i<-2$,
$F_{i} = 0$ for all $i>0$ and
\begin{equation*}
  F_{-2} =\fh= \fh_{-2} \oplus \fh_{-1} \oplus \fh_{0}~,\qquad
  F_{-1} = \fh_{-1} \oplus \fh_{0}~,\qquad F_0 = \fh_0~.
\end{equation*}
The Lie superalgebra structure on $F$ satisfies $[F_i,F_j] \subset
F_{i+j}$ and we are interested in those structures such that the
components of the Lie brackets of zero filtration degree coincide with
the Lie brackets of $\fh$.

For the Lie superalgebras of interest we can be very concrete and
describe the most general filtered deformation of $\fh$ by the
following brackets:
\begin{equation*}
  \begin{aligned}[m]
	  & [\fh_{0}, V] & \subset  V\oplus\fh_{0}\;\\
    & [S, S] & \subset  V\oplus \fh_{0}\;\\
    & [V, S] & \subset S\phantom{\oplus\fh_0c}\,\;\\
    & [V, V] & \subset V \oplus \fh_{0}\;
  \end{aligned}
\end{equation*}
and the condition that the associated graded Lie superalgebra should
be isomorphic to $\fh$ translates into the condition that the
component in $V$ of the brackets $[\fh_{0},V]$
and $[S,S]$ should
not be modified from the ones in the Poincaré superalgebra.
The components of the Lie brackets of non-zero filtration degree are
as follows:

\begin{enumerate}[label=(\roman*)]
\item the even component $\mu$ is the sum
  $\mu=\alpha+\beta+\gamma+\rho$ of the degree-$2$ maps
  \begin{equation} 
    \label{eq:filteredef}
    \begin{split}
      \alpha &:\Lambda^2 V\to V~,\qquad\beta:V\otimes S\to S~,\\
      \gamma &:\odot^2 S\to \fh_0~,\qquad\rho:\fh_0\otimes V\to\fh_0~;
    \end{split}
  \end{equation}
\item and the even component $\delta:\Lambda^2 V\to \fh_0$ of degree $4$.
\end{enumerate}

Calculating the deformations involves, at first order, the
calculation of the cohomology of an appropriate refinement of the
Chevalley--Eilenberg complex which we now describe.  It is a
refinement (by degree) $H^{\bullet,\bullet}(\fm,\fg)$ of the usual
Chevalley-Eilenberg cohomology $H^\bullet(\fg,\fg)$ associated with a
Lie (super)algebra $\fg$ and its adjoint representation to the case of
$\ZZ$-graded Lie (super)algebras $\fg=\bigoplus_{j\in\mathbb Z}\fg_j$
with negatively graded part $\fm=\bigoplus_{j<0}\fg_{j}$.  We will
consider first the case of the full Poincaré superalgebra $\fg$.

\subsection{The Spencer complex}
\label{sec:spencer-complex}

The cochains of the Spencer complex are even linear maps
$\Lambda^p \fm \to \fg$ or, equivalently, even elements of
$\fg \otimes \Lambda^p \fm^*$, where $\Lambda^\bullet$ is meant here
in the super sense. One extends the degree in $\fg$ to such cochains
by declaring that $\fg_j^*$ has degree $-j$. Since the $\ZZ$- and
$\ZZ_2$ gradings are compatible, even (resp. odd) cochains have even
(resp. odd) degree. It is not hard to see that the even $p$-cochains
of highest degree are the maps $\Lambda^p V \to \fso(V)$, which have
degree $2p$. The even $p$-cochains of lowest degree are those in
$\Hom(\odot^p S, V)$, for $p \equiv 0 \pmod 2$, which have degree
$p-2$, and those in $\Hom(\odot^p S, S)$ and
$\Hom(\odot^{p-1} S \otimes V, V)$, for $p \equiv 1 \pmod 2$, which
have degree $p-1$. As we will see below, the Spencer differential has
degree $0$, so the complex breaks up into a direct of sum of
\emph{finite} complexes for each degree.

The spaces in the complexes of even cochains for small degree are
given in Table~\ref{tab:even-cochains-small}; although the complex in
degree $4$ has cochains also for $p=5,6$ which the table omits.
We shall be mainly interested in $p=2$ in this paper, which, as we
will see in Theorem~\ref{thm:first} later on, corresponds to
infinitesimal  deformations.

\begin{table}[h!]
  \centering
  \caption{Even $p$-cochains of small degree}
  \label{tab:even-cochains-small}
  \begin{tabular}{c*{5}{|>{$}c<{$}}}
    \multicolumn{1}{c}{} & \multicolumn{5}{|c}{$p$} \\\hline
    deg & 0 & 1 & 2 & 3 & 4 \\\hline
    0 & \fso(V) & \begin{tabular}{@{}>{$}c<{$}@{}} S \to S\\ V \to
                    V\end{tabular} & \odot^2 S \to V & & \\\hline
    2 & & V \to \fso(V) & \begin{tabular}{@{}>{$}c<{$}@{}}
                            \Lambda^2 V \to V\\ V \otimes S \to S \\
                            \odot^2 S \to \fso(V) \end{tabular}
         & \begin{tabular}{@{}>{$}c<{$}@{}} \odot^3 S \to S\\
             \odot^2 S \otimes V \to V\end{tabular} & \odot^4 S \to V
        \\\hline
    4 & & & \Lambda^2 V \to \fso(V) & \begin{tabular}{@{}>{$}c<{$}@{}}
   \odot^2 S \otimes V \to \fso(V) \\
   \Lambda^2 V \otimes S \to S \\ 
   \Lambda^3 V \to V \end{tabular} & \begin{tabular}{@{}>{$}c<{$}@{}}
  \odot^4 S \to \fso(V) \\
  \odot^3 S \otimes V \to S \\
  \odot^2 S\otimes \Lambda^2 V\to V \end{tabular} \\\hline
  \end{tabular}
\end{table}

Let $C^{d,p}(\fm,\fg)$ denote the space of $p$-cochains of degree
$d$.  The Spencer differential $\partial: C^{d,p}(\fm,\fg) \to
C^{d,p+1}(\fm,\fg)$ is the Chevalley--Eilenberg differential for the
Lie superalgebra $\fm$ relative to its module $\fg$ with respect to
the adjoint action.  For $p=0,1,2$ and $d\equiv 0 \pmod 2$ it is
explicitly given  by the following expressions:
\begin{align}
  \begin{split}\label{eq:Spencer0}
    &\partial : C^{d,0}(\fm,\fg) \to C^{d,1}(\fm,\fg)\\
    &\partial\varphi(X) = [X,\varphi]~,
  \end{split}
  \\
  \begin{split}\label{eq:Spencer1}
    &\partial : C^{d,1}(\fm,\fg)\to C^{d,2}(\fm,\fg)\\
    &\partial\varphi(X,Y) = [X,\varphi(Y)] - (-1)^{xy} [Y,\varphi(X)] - \varphi([X,Y])~,
  \end{split}
\\
  \begin{split}\label{eq:Spencer2}
    &\partial:C^{d,2}(\fm,\fg) \to C^{d,3}(\fm,\fg)\\
    &\partial\varphi(X,Y,Z) = [X,\varphi(Y,Z)]+(-1)^{x(y+z)}[Y,\varphi(Z,X)] + (-1)^{z(x+y)} [Z,\varphi(X,Y)] \\
    & {} \qquad\qquad\qquad - \varphi([X,Y],Z) - (-1)^{x(y+z)} \varphi([Y,Z],X) -(-1)^{z(x+y)} \varphi([Z,X],Y)~,
  \end{split}
\end{align}
where $x,y,\dots$ denote the parity of elements $X,Y,\dots$ of $\fm$
and $\varphi\in C^{d,p}(\fm,\fg)$ with $p=0,1,2$ respectively.

The space of cochains $C^{d,p}(\fm,\fg)$ is an $\fso(V)$-module and
the same is true for the spaces of cocycles and coboundaries, as
$\partial$ is $\fso(V)$-equivariant; this implies that each cohomology
group $H^{d,p}(\fm,\fg)$ is an $\fso(V)$-module, in a natural
way. This equivariance is very useful in calculations, as we will have
ample opportunity to demonstrate.

Many of the components of the Spencer differential turn out to be
injective. For instance, for all $\varphi\in \Hom(\Lambda^2 V,
\fso(V))$ in degree $4$, one has
\begin{align*}
  &\partial\varphi(s_1,s_2,v_1)=-\varphi([s_1,s_2],v_1)\\
  &\partial\varphi(v_1,v_2,s_1)=[s_1,\varphi(v_1,v_2)]
\end{align*}
where $s_1,s_2\in S$ and $v_1,v_2\in V$ and the two components
\begin{align*}
  &\Hom(\Lambda^2 V, \fso(V)) \rightarrow \Hom(\odot^2S \otimes
  V,\fso(V))\\
  &\Hom(\Lambda^2 V, \fso(V)) \rightarrow \Hom(\Lambda^2V \otimes
  S, S)
\end{align*}
of $\partial: C^{4,2}(\fm,\fg) \to C^{4,3}(\fm,\fg)$ are injective (in
the first case one uses that $\fm$ is fundamental).  We also note for
completeness that the third component is surjective but has nonzero
kernel, giving rise to the short exact sequence
\begin{equation*}
  0 \rightarrow V^{\yng(2,2)} \rightarrow \Hom(\Lambda^2 V, \fso(V))
  \rightarrow \Hom(\Lambda^3 V, V) \rightarrow 0~,
\end{equation*}
where $V^{\yng(2,2)}$ is the space of algebraic curvature operators;
that is, the subspace of $S^2\Lambda^2 V$ satisfying the algebraic
Bianchi identity.  One has the following

\begin{lemma}
  \label{lem:p4}
  The group $H^{d,2}(\fm,\fg)=0$ for all even $d\geq 4$.
\end{lemma}

\begin{proof}
  If $d=4$ then $\operatorname{Ker}\partial|_{C^{4,2}(\fm,\fg)}=0$
  from the previous observations; if $d>4$ then the space of cochains
  $C^{d,2}(\fm,\fg)=0$ and the claim is immediate.
\end{proof}

In degree $2$ it is convenient to consider the decomposition of
$\fso(V)$-modules
\begin{equation*}
C^{2,2}(\fm,\fg)=\Hom(\Lambda^2 V, V)\oplus \Hom(V \otimes S, S)
\oplus \Hom(\odot^2S, \fso(V))
\end{equation*}
and the corresponding $\fso(V)$-equivariant projections
\begin{equation}
  \label{eq:projectors}
  \begin{split}
    \pi^{\alpha} &: C^{2,2}(\fm,\fg)\to \Hom(\Lambda^2 V, V)\\
    \pi^{\beta} &: C^{2,2}(\fm,\fg)\to \Hom(V \otimes S, S)\\
    \text{and}\qquad\pi^{\gamma} &: C^{2,2}(\fm,\fg)\to \Hom(\odot^2S,
    \fso(V))~.
  \end{split}
\end{equation}
We find that
\begin{equation*}
  \partial\varphi(s_1,s_2)=-\varphi([s_1,s_2]) \qquad\text{and}\qquad
  \partial\varphi(v_1,s_1)=[\varphi(v_1),s_1]
\end{equation*}
for all $\varphi\in \Hom(V,\fso(V))$, and that two of the three
components of $\partial : C^{2,1}(\fm,\fg) \to C^{2,2}(\fm,\fg)$ are
injective:
\begin{align*}
  &\Hom(V,\fso(V)) \hookrightarrow \Hom(\odot^2S, \fso(V)) \\
  &\Hom(V,\fso(V)) \hookrightarrow \Hom(V \otimes S, S)\;.
\end{align*}
On the other hand, the image of $\varphi$ under $\partial^\alpha$ is given by
  \begin{equation*}
    \partial^\alpha\varphi(v_1,v_2) = [v_1, \varphi(v_2)] -
    [v_2,\varphi(v_1)]~,
  \end{equation*}
  for $v_1,v_2 \in V$. This easily implies the following
\begin{lemma}
  \label{lem:iso}
  The component
  \begin{equation*}
    \partial^\alpha:=\pi^\alpha\circ\partial : \Hom(V,\fso(V)) \to \Hom(\Lambda^2 V, V)
  \end{equation*}
  of the Spencer differential is an isomorphism.
\end{lemma}

\section{Infinitesimal deformations}
\label{sec:infin-deform}

In this section we first calculate the cohomology group
\begin{equation*}
  H^{2,2}(\fm,\fg) = \frac{\ker\partial: C^{2,2}(\fm,\fg) \to
    C^{2,3}(\fm,\fg)}{\partial C^{2,1}(\fm,\fg)}~,
\end{equation*}
and then consider the $\mathbb Z$-graded Lie subalgebras $\fh$ of the
Poincaré superalgebra. Using the results obtained for $\fg$, we will
describe  the groups $H^{d,2}(\fm,\fh)$ for all $d\geq 2$ even and
then prove Theorem~\ref{thm:first} about the infinitesimal
deformations of $\fh$.

\subsection{Infinitesimal deformations of \texorpdfstring{$\fg$}{g}}
\label{sec:infin-deform-g}

We depart from the following observation.

\begin{lemma}\label{lem:cocyclerep}
  Every cohomology class $[\alpha + \beta + \gamma] \in
  H^{2,2}(\fm,\fg)$ with $\alpha \in \Hom(\Lambda^2V,V)$, $\beta \in
  \Hom(V \otimes S, S)$ and $\gamma \in \Hom(\odot^2S, \fso(V))$, has a
  unique cocycle representative with $\alpha = 0$.
\end{lemma}

\begin{proof}
  It follows from Lemma~\ref{lem:iso} that given any
  $\alpha \in \Hom(\Lambda^2V,V)$, there is a unique
  $\tilde\alpha \in C^{2,1}(\fm,\fg)$ such that
  $\partial \tilde\alpha = \alpha + \tilde\beta + \tilde\gamma$, for
  some $\tilde\beta \in \Hom(V \otimes S, S)$ and some
  $\tilde\gamma \in \Hom(\odot^2S,\fso(V))$.  Therefore given any
  cocycle $\alpha + \beta + \gamma$ we may add the coboundary
  $\partial(-\tilde\alpha)$ without changing its cohomology class,
  resulting in the cocycle $(\beta - \tilde\beta) + (\gamma -
  \tilde\gamma)$, which has no component in $\Hom(\Lambda^2 V ,V)$.
\end{proof}

In other words, $H^{2,2}(\fm,\fg)$ is isomorphic as an
$\fso(V)$-module to the kernel of the Spencer differential restricted
to $\Hom(V\otimes S, S) \oplus \Hom(\odot^2S,\fso(V))$.  It follows
from equation~\eqref{eq:Spencer2} for the Spencer differential, that a
cochain $\beta + \gamma \in \Hom(V\otimes S, S) \oplus
\Hom(\odot^2S,\fso(V))$ is a cocycle if and only if the following pair of
``cocycle conditions'' are satisfied:
\begin{equation}
  \label{eq:cc1}
  [\gamma(s,s),v] = - 2 [s,\beta(v,s)] \qquad\forall s\in S, v\in V~,
\end{equation}
and
\begin{equation}
  \label{eq:cc2}
  [\gamma(s,s),s] = -\beta([s,s],s) \qquad \forall s\in S~.\qquad\quad~
\end{equation}
We note that the cocycle condition \eqref{eq:cc1} fully expresses
$\gamma$ in terms of $\beta$, once the fact that $\gamma$ takes values
into $\fso(V)$ has been taken into account. To this aim, we define for
any $v\in V$ the endomorphism $\beta_v \in \End(S)$ by
$\beta_v(s) = \beta(v,s)$ and rewrite \eqref{eq:cc1} as
\begin{equation*}
  [\gamma(s,s),v] = -2 [s,\beta_v(s)]~.
\end{equation*}
Take the inner product with $v$ and use equations
\eqref{eq:DiracCurrent} and \eqref{eq:symplectic} to arrive at
\begin{equation}
  \label{eq:symmetric-endo}
  0 = 2 \left<s, v \cdot \beta_v(s)\right> \qquad \forall s\in S, v\in V~.
\end{equation}
This says that for all $v\in V$ the endomorphism $v\cdot \beta_v$ of
$S$ is symmetric relative to the symplectic form
$\left<-,-\right>$.  Equivalently, it is fixed by the anti-involution
$\sigma$ defined by the symplectic form: $\sigma(v\cdot \beta_v) = v
\cdot \beta_v$.

We now observe that if $\Theta \in \Cl(V)$ is fixed by $\sigma$, then
so are $v \cdot \Theta \cdot v$ and (trivially) $v \cdot v \cdot
\Theta = -\eta(v,v)\Theta$, so that we have an immediate class of
solutions to equation~\eqref{eq:symmetric-endo}: namely, $\beta_v = v
\cdot \Theta + \Theta' \cdot v$, where $\Theta,\Theta'\in \Cl(V)$ are
fixed by $\sigma$.

Following our conventions on Clifford algebras and spinors in Appendix
\ref{sec:clifford-conventions}, we have
\begin{equation}
\label{eq:decomposition}
  \End(S) \cong \bigoplus_{p=0}^5 \Lambda^p V~
\end{equation}
as $\fso(V)$-modules.  The anti-involution $\sigma$ preserves each
submodule and acts on the submodule of
type $\Lambda^p V$ as $(-1)^{p(p+1)/2}\1$, so that the submodule fixed
by $\sigma$ correspond to $\Lambda^2S \cong \Lambda^0 V \oplus
\Lambda^3 V \oplus \Lambda^4 V$.  In other words,
equation~\eqref{eq:symmetric-endo} says that for all $v\in V$,
\begin{equation}
  \label{eq:cc3}
  v \cdot \beta_v \in \Lambda^0V \oplus \Lambda^3V \oplus \Lambda^4V~;
\end{equation}
that is, strictly speaking, in the image of $\Lambda^0V \oplus
\Lambda^3V \oplus \Lambda^4V$ in $\End(S)$.

As we have seen above, we can exhibit solutions to
equation~\eqref{eq:cc3} of the form
\begin{equation*}
  \beta_v = v \cdot \Theta + \Theta' \cdot v~,
\end{equation*}
for $\Theta, \Theta' \in \Lambda^0 V \oplus \Lambda^3 V \oplus
\Lambda^4 V$.  Remarkably, it turns out that these are all the
solutions to equation~\eqref{eq:cc3}.

\begin{proposition}
  \label{prop:solutioncc3}
  The general solution of equation~\eqref{eq:cc3} is
  \begin{equation*}
    \beta_v = \theta_0 v + v \cdot \theta_3 + \theta'_3 \cdot v + v
    \cdot \theta_4 + \theta'_4 \cdot v~,
  \end{equation*}
 where $\theta_0 \in \Lambda^0V$, $\theta_3,\theta'_3 \in \Lambda^3
  V$ and $\theta_4,\theta'_4 \in \Lambda^4V$.
\end{proposition}

Although a more combinatorial proof of
Proposition~\ref{prop:solutioncc3} is also possible, we give here a
proof which uses representation theory and the $\fso(V)$-equivariance
of the condition~\eqref{eq:cc3}. To do so, we will use freely the
notation in Appendix~\ref{sec:some-representations} and identify the
$\fso(V)$-modules $\Hom(V\otimes S, S)$ and $\Hom(V,\End(S))$.

We start by reformulating slightly Proposition~\ref{prop:solutioncc3}. 
Let 
\begin{equation*}
  \Phi: \Hom(V,\End(S)) \to \Hom(\odot^2V,\End(S))
\end{equation*}
denote the $\fso(V)$-equivariant map which sends $\beta \in
\Hom(V,\End(S))$ to $\Phi(\beta)$, given for all $v,w \in V$ by
\begin{equation}
  \label{eq:equivariant} 
  \Phi(\beta)(v,w) = v\cdot \beta_w + w \cdot \beta_v~,
\end{equation}
where $\cdot$ stands, as usual, for the Clifford product. We start
with a useful observation.

\begin{lemma}
  \label{le:injective}
  The map $\Phi:\Hom(V,\End(S)) \to \Hom(\odot^2V,\End(S))$ is
  injective.
\end{lemma}

\begin{proof}
  Suppose that $\Phi(\beta) =0$.  This means that for all $v \in V$,
  $v \cdot \beta_v = 0$.  By Clifford-multiplying on the left with
  $v$, we learn that $\beta_v = 0$ for all $v \in V$ with $\eta(v,v)
  \neq 0$.  But $v \mapsto \beta_v$ is linear and there exists a basis
  for $V$ whose elements have nonzero norm, hence $\beta_v = 0$ for
  all $v\in V$ and hence $\beta = 0$.
\end{proof}

Using the $\fso(V)$-module decomposition \eqref{eq:decomposition}
\begin{equation*}
  \Hom(V,\End(S)) \cong \bigoplus_{p=0}^5 \Hom(V,\Lambda^p V)~,
\end{equation*}
we may decompose $\beta = \beta_0 + \beta_1 + \cdots + \beta_5$, where
$\beta_p$ belongs to the $\fso(V)$-submodule $\Hom(V,\Lambda^p V)$ of
$\Hom(V,\End(S))$.  Similarly we have an $\fso(V)$-equivariant
isomorphism
\begin{equation*}
  \Hom(\odot^2 V,\End(S)) \cong \bigoplus_{q=0}^5
  \Hom(\odot^2V,\Lambda^q V)
\end{equation*}
and a corresponding decomposition
$\theta = \theta_0 + \theta_1 + \cdots +
\theta_5$ of $\theta \in \Hom(\odot^2 V,\End(S))$, with $\theta_q$ belonging to the
$\fso(V)$-submodule $\Hom(\odot^2 V, \Lambda^q V)$
of $\Hom(\odot^2 V, \End(S))$.

We now observe that equation~\eqref{eq:cc3} for $\beta$, which 
says that $\Phi(\beta)(v,w)$ is a symmetric endomorphism of $S$ for all $v,w\in V$,
is equivalent to $\Phi(\beta)_q = 0$ for $q=1,2,5$ and recall that the
solution space of these three equations contains a submodule of type
\begin{equation}
  \label{eq:solutionmodule}
  \Lambda^0 V \oplus 2 \Lambda^3 V \oplus 2 \Lambda^4 V~.
\end{equation}
From Table~\ref{tab:betacomps}, which lists the decomposition of
$\Hom(V,\Lambda^p V)$, for $p=0,1,\dots,5$, into irreducible
$\fso(V)$-modules, we see that there is a unique $\fso(V)$-submodule
isomorphic to \eqref{eq:solutionmodule}, whose irreducible components
appear inside boxes.  As explained in
Appendix~\ref{sec:some-representations}, the notation
$(V \otimes \Lambda^p V)_0$ stands for the kernel of Clifford
multiplication.

\begin{table}[h!]
  \centering
  \caption{Irreducible components of $\Hom(V,\Lambda^pV)$ for
    $p=0,\dots,5$.}
  \label{tab:betacomps}
  \begin{tabular}{>{$}c<{$}|>{$}l<{$}}
    \multicolumn{1}{c|}{p} & \multicolumn{1}{c}{$\Hom(V,\Lambda^p V)$}\\\hline
    0 & V\\
    1 & \boxed{\Lambda^0 V} \oplus \Lambda^2 V \oplus \odot^2_0 V\\
    2 & V \oplus \boxed{\Lambda^3 V} \oplus (V \otimes
        \Lambda^2V)_0\\
    3 & \Lambda^2 V \oplus \boxed{\Lambda^4 V} \oplus (V \otimes
        \Lambda^3V)_0\\
    4 & \boxed{\Lambda^3 V} \oplus \Lambda^5 V \oplus (V \otimes
        \Lambda^4V)_0\\
    5 & \boxed{\Lambda^4 V} \oplus \Lambda^5 V \oplus (V \otimes
        \Lambda^5V)_0
  \end{tabular}
\end{table}

It follows from this discussion that
Proposition~\ref{prop:solutioncc3} is equivalent to the following.

\begin{proposition}
\label{prop:first-coc-equivalent}
  The solution space of equation~\eqref{eq:cc3} is the unique
  submodule of $\Hom(V,\End(S))$ isomorphic to
  \eqref{eq:solutionmodule}.
\end{proposition}

\begin{proof}
  It follows from the first formula in \eqref{eq:VonForms} that
  Clifford multiplication maps $V \otimes \Lambda^p V \to
  \Lambda^{p-1}V \oplus \Lambda^{p+1}V$.  From this fact and the very
  definition \eqref{eq:equivariant} of the map $\Phi$, it is clear
  that $\Phi(\beta_p)_q = 0$ unless $q = p \pm 1$.  Therefore
  equation~\eqref{eq:cc3} is equivalent to the following
  system of linear equations:
  \begin{equation}
    \label{eq:threeequations}
    \begin{split}
      \Phi(\beta)_1 = 0 \iff \Phi(\beta_0)_1 + \Phi(\beta_2)_1 &= 0\\
      \Phi(\beta)_2 = 0 \iff \Phi(\beta_1)_2 + \Phi(\beta_3)_2 &= 0\\
      \Phi(\beta)_5 = 0 \iff \Phi(\beta_4)_5 + \Phi(\beta_5)_5 &= 0~.
    \end{split}
  \end{equation}
  Note that each component $\beta_0,\ldots,\beta_5$ of $\beta$ appear
  in one and only one of the above three equations.
  Table~\ref{tab:thetacomps} lists the irreducible
  $\fso(V)$-components in $\Hom(\odot^2 V,\Lambda^q V)$ for
  $q=0,\dots,5$ which are isomorphic to one of the irreducible modules
  appearing in Table \ref{tab:betacomps}.  By $\fso(V)$-equivariance
  the image of $\Phi$ is isomorphic to the direct sum of the
  irreducible modules displayed in Table \ref{tab:betacomps} and is
  contained in the direct sum of the irreducible modules displayed in
  Table \ref{tab:thetacomps}.

  \begin{table}[h!]
    \centering
    \caption{Some irreducible components of $\Hom(\odot^2 V,\Lambda^q V)$ for
      $q=0,\dots,5$.}
    \label{tab:thetacomps}
    \begin{tabular}{>{$}c<{$}|>{$}l<{$}}
      \multicolumn{1}{c|}{q} & \multicolumn{1}{c}{$\Hom(\odot^2 V,\Lambda^q V)$}\\\hline
      0 & \Lambda^0 V\oplus \odot
          ^2_0 V\\
      1 & 2V \oplus (V \otimes
          \Lambda^2V)_0\\
      2 & 2\Lambda^2 V \oplus \odot^2_0 V \oplus (V \otimes
          \Lambda^3V)_0\\
      3 & 2\Lambda^3 V \oplus (V \otimes
          \Lambda^4V)_0\oplus (V \otimes
          \Lambda^2V)_0 \\
      4 & 2\Lambda^4 V\oplus (V \otimes
          \Lambda^5V)_0\oplus (V \otimes
          \Lambda^3V)_0\\
      5 & 2\Lambda^5 V \oplus (V \otimes
          \Lambda^5V)_0\oplus (V \otimes
          \Lambda^4V)_0
    \end{tabular}
  \end{table}

  Let $\Phi^p_q$ denote the component of $\Phi$ mapping
  $\Hom(V,\Lambda^p V)$ to $\Hom(\odot^2 V, \Lambda^q V)$.  We will
  have proved the proposition if we show that the undesirable
  irreducible components of $\Hom(V,\End(S))$ (those not boxed in
  Table~\ref{tab:betacomps}) are not in the kernel of $\Phi^\bullet_q$
  for $q=1,2,5$.  Since each $\Phi^p_q$ is $\fso(V)$-equivariant, it
  is enough to show that this is the case for each type of undesirable
  submodule.  We now go through each such submodule in turn.

  Let $2V$ be the isotypical component of $V$ inside
  $\Hom(V,\End(S))$; it is contained in
  $\Hom(V,\Lambda^0V)\oplus\Hom(V,\Lambda^2 V)$.  From
  Table~\ref{tab:thetacomps} and Lemma~\ref{le:injective}, $\Phi$ maps
  $2V$ injectively into $\Hom(\odot^2 V,\Lambda^1 V)$.  It
  follows that $\Phi|_{2V}=(\Phi^0_1+\Phi^2_1)|_{2V}:2V\to\Hom(\odot^2
  V,\Lambda^1 V)$ is injective, and thus the first equation in
  \eqref{eq:threeequations} is not satisfied by any nonzero
  $\beta=\beta_0+\beta_2\in 2V$ and the solution space of
  equation~\eqref{eq:cc3} does not contain any submodule  isomorphic to $V$.
	
  A similar argument shows that the isotypical components
  $2\Lambda^2 V$ and $2\Lambda^5V$ are mapped injectively to
  submodules of $\Hom(\odot^2 V,\Lambda^2 V)$ and $\Hom(\odot^2
  V,\Lambda^5 V)$, respectively, and hence the solution space of
  equation~\eqref{eq:cc3} does not contain any submodule isomorphic to
  $\Lambda^2 V$ or $\Lambda^5V$ either.

  The remaining submodules are isomorphic to
  $(V \otimes \Lambda^p V)_0$ for $p=1,2,3,4,5$ and are unique.  Since
  any such module is irreducible, the equivariant map $\Phi$ is either
  zero or an isomorphism when restricted to it.  Thus we find it
  easiest to pick a nonzero element and show that its image under the
  relevant component of $\Phi$ is nonzero.

  Let $\beta = \be^\flat_1 \otimes \be_2 + \be^\flat_2 \otimes
  \be_1$.  It belongs to the submodule of type $(V\otimes V)_0 \cong
  \odot^2_0 V$ and  a calculation in $\Cl(V)$ shows that
  \begin{equation*}
    \Phi(\beta)_2(\be_1,\be_1) = -2\be_1 \wedge \be_2 \neq 0~.
  \end{equation*}

  Let $\beta = \be^\flat_1 \otimes \be_2 \wedge \be_3 +
  \be^\flat_2 \otimes \be_1 \wedge \be_3$.  It belongs to the
  submodule of type $(V \otimes \Lambda^2V)_0$ and
  \begin{equation*}
    \Phi(\beta)_1(\be_1 + \be_2,\be_1+\be_2) = -4 \be_3 \neq 0~.
  \end{equation*}

  Let $\beta = \be^\flat_1\otimes \be_2 \wedge\be_3 \wedge \be_4 +
  \be^\flat_2\otimes \be_1 \wedge\be_3 \wedge \be_4$. It belongs to the
  submodule of type $(V \otimes \Lambda^3 V)_0$ and
  \begin{equation*}
    \Phi(\beta)_2(\be_1 + \be_2, \be_1+\be_2) = -4 \be_3 \wedge
    \be_4 \neq 0~.
  \end{equation*}

  Let $\beta = \be^\flat_1 \otimes \be_2 \wedge\be_3 \wedge \be_4
  \wedge \be_5 + \be^\flat_2 \otimes \be_1 \wedge\be_3 \wedge \be_4
  \wedge \be_5$. It belongs to the submodule of type $(V\otimes \Lambda^4
  V)_0$ and
  \begin{equation*}
    \Phi(\beta)_5(\be_1,\be_1) = -2\be_1 \wedge \be_2 \wedge \be_3
    \wedge \be_4 \wedge \be_5 \neq 0~.
  \end{equation*}

  Finally, let $\beta = \be^\flat_1 \otimes \be_2 \wedge\be_3 \wedge \be_4
  \wedge \be_5 \wedge \be_6 + \be^\flat_2 \otimes \be_1 \wedge\be_3 \wedge \be_4
  \wedge \be_5\wedge \be_6$. It belongs to the submodule of type $(V\otimes \Lambda^5
  V)_0$ and
  \begin{equation*}
    \Phi(\beta)_5(\be_1,\be_1) = -2\be_1 \wedge \be_2 \wedge \be_3
    \wedge \be_4 \wedge \be_5 \wedge \be_6 \neq 0~,
  \end{equation*}
  thought as an element of $\Lambda^6 V\simeq \Lambda^5 V$ in
  $\End(S)$. This concludes the proof of the proposition.
\end{proof}

Returning to the cocycle conditions \eqref{eq:cc1} and \eqref{eq:cc2},
we now observe that the first one simply defines $\gamma$ in terms of
$\beta$, which is then subject to the second condition.  Given the
general form of $\beta_v$ found in Proposition~\ref{prop:solutioncc3},
we solve the cocycle conditions in the following

\begin{proposition}
  \label{prop:beta}
  The general solution $(\beta,\gamma)$ of the cocycle conditions
  \eqref{eq:cc1} and \eqref{eq:cc2} is of the form
  $(\beta,\gamma)=(\beta^\varphi,\gamma^\varphi)$ for a unique
  $\varphi \in \Lambda^4V$ such that
  \begin{equation} 
    \label{eq:beta}
    \begin{split}
      \beta^\varphi_v(s)&= v \cdot \varphi\cdot s - 3\, \varphi \cdot v\cdot s~,\\
      [\gamma^\varphi(s,s),v]&=-2[s,\beta^\varphi_v(s)]
      \\&=-2[s,v\cdot\varphi\cdot s]+6[s,\varphi\cdot v \cdot s]~,
    \end{split}
  \end{equation}
  for all $v\in V$ and $s\in S$.
  In particular $H^{2,2}(\fm,\fg)\simeq \Lambda^4 V$ as an $\fso(V)$-module.
\end{proposition}

\begin{remark}
  Expanding the Clifford products, we may rewrite $\beta^\varphi_v$ as
  \begin{equation*}
    \beta^\varphi_v = -2 v \wedge \varphi - 4 \iota_v \varphi~,
  \end{equation*}
  which agrees with the zeroth order terms in the connection $D$ in
  equation \eqref{eq:D} for $\varphi = \tfrac1{24}F$.  The connection
  $D$ encodes the geometry of (supersymmetric) bosonic backgrounds of
  eleven-di\-men\-sional supergravity: not just does it define the notion
  of a Killing spinor, but its curvature encodes the bosonic field
  equations.  Indeed, as shown in \cite{GauPak}, the field equations
  are precisely the vanishing of the gamma-trace of the curvature of
  $D$.  Proposition~\ref{prop:beta} can be paraphrased as showing that
  we are able to reconstruct eleven-di\-men\-sional supergravity (at least
  at the level of the bosonic field equations) from the Spencer
  cohomology of the Poincaré superalgebra.
\end{remark}

\begin{proof}
Rewriting equation~\eqref{eq:cc2} as
\begin{equation*}
  [\gamma(s,s),s] + \beta_{[s,s]}(s) = 0~,
\end{equation*}
with $\gamma$ given in terms of $\beta$ by equation~\eqref{eq:cc1}, we
see that its solutions $\beta$ are the kernel of an
$\fso(V)$-equivariant linear map.  The kernel consists of submodules
and hence it is enough, given Proposition~\ref{prop:first-coc-equivalent}, to
study this equation separately for $\beta$ belonging to an isotypical
component of type $\Lambda^0 V$, $2\Lambda^3V$ and $2\Lambda^4V$,
respectively.

It is convenient in what follows to work in $\Cl(V)$.  This uses the
notation explained in Appendix~\ref{sec:clifford-conventions} and the
Einstein summation convention.

Let us define $\gamma(s,s)^\mu{}_\nu$ by
\begin{equation*}
  \gamma(s,s)(\be_\nu) = \gamma(s,s)^\mu{}_\nu \be_\mu~.
\end{equation*}
It follows from the first cocycle condition \eqref{eq:cc1} that
\begin{equation}
  \label{eq:gammabasis}
  \gamma(s,s)_{\mu\nu} = 2 \sbar \Gamma_\mu \beta_\nu s~,
\end{equation}
where have abbreviated $\beta_{\be_\nu}$ by $\beta_\nu$.  Using
equation~\eqref{eq:spingens}, the image of $\gamma(s,s)$ in $\Cl(V)$
is given by
\begin{equation*}
  \gamma(s,s) \mapsto -\tfrac12 (\sbar \Gamma_\mu \beta_\nu s)
  \Gamma^{\mu\nu}~,
\end{equation*}
and hence
\begin{equation*}
  [\gamma(s,s),s] = -\tfrac12 (\sbar \Gamma_\mu \beta_\nu s)
  \Gamma^{\mu\nu} s~.
\end{equation*}
The second term of the second cocycle condition \eqref{eq:cc2} is
given by
\begin{equation*}
  \beta_{[s,s]} s = - (\sbar \Gamma^\mu s) \beta_\mu s~,
\end{equation*}
so that the second cocycle condition becomes
\begin{equation}
  \label{eq:cc5}
  \tfrac12 (\sbar \Gamma_\mu \beta_\nu s) \Gamma^{\mu\nu} s +
  (\sbar \Gamma^\mu s) \beta_\mu s = 0~.
\end{equation}
It is enough to consider three different cases for $\beta$.

Let $\beta$ be of type $\Lambda^0 V$, so that
$\beta_\mu = \theta_0 \Gamma_\mu$ for some
$\theta_0 \in \Lambda^0 V$.  Then equation~\eqref{eq:cc5} becomes
\begin{equation}
  \label{eq:cc5b0}
  \tfrac12 \theta_0 (\sbar \Gamma^{\mu\nu} s) \Gamma_{\mu\nu} s = 
  -\theta_0 (\sbar \Gamma^\mu s) \Gamma_\mu s~.
\end{equation}
From the first of equations~\eqref{eq:trilinears},
equation~\eqref{eq:cc5b0} becomes
\begin{equation*}
  \theta_0 (\sbar \Gamma^\mu s) \Gamma_\mu s = 0~. 
\end{equation*}
Taking the symplectic inner product with $s$ we find
\begin{equation*}
  \theta_0 (\sbar \Gamma^\mu s)(\sbar \Gamma_\mu s) = \theta_0
  \eta([s,s],[s,s])=0
\end{equation*}
for all $s\in S$, where $[s,s]$ is the Dirac current of $s$. On the
other hand there always exists an $s$ for which $\eta([s,s],[s,s]) >
0$ and hence the only way \eqref{eq:cc5b0} is satisfied for all $s$
is if $\theta_0 = 0$.

The next two cases, $\beta$ in $2\Lambda^3V$ and $2\Lambda^4 V$,
are computationally more involved.  It pays to exploit the
equivariance under $\fso(V)$, which implies first of all that the
solution space is an $\fso(V)$-module.  We also notice that both
$\Lambda^3 V$ and $\Lambda^4 V$ are real irreducible representations
of $\fso(V)$ which remain irreducible upon complexification.  This
means that the only $\fso(V)$-equivariant endomorphisms of $\Lambda^3
V$ and $\Lambda^4 V$ are real multiples of the identity.

Now suppose that $\beta$ is in $2\Lambda^3 V$.  The solution
space to equation~\eqref{eq:cc5} is an $\fso(V)$-submodule of
$\Lambda^3 V \oplus \Lambda^3 V$, hence it is either all of
$2\Lambda^3 V$ (which happens if the equations are identically zero),
or a copy of $\Lambda^3 V$ given by the image of
\begin{equation*}
  \Lambda^3 V \ni \psi \mapsto (t_1 \psi, t_2 \psi) \in \Lambda^3 V
  \oplus \Lambda^3 V~,
\end{equation*}
for some $t_1,t_2 \in \RR$.  We also allow for the case of a
zero-di\-men\-sional solution space when $t_1=t_2=0$.  We put
$\beta_\mu = t_1 \Gamma_\mu \psi + t_2 \psi \Gamma_\mu$, for
$\psi \in \Lambda^3 V$, into equation~\eqref{eq:cc5} to arrive at the
following set of equations for $t_1,t_2\in\RR$:
\begin{equation}
  \label{eq:cc5b3}
  0 = t_1 \left(\tfrac12 (\sbar \Gamma_{\mu\nu} \psi s) \Gamma^{\mu\nu}
    s + (\sbar \Gamma^\mu s) \Gamma_\mu \psi s \right) + t_2
  \left(\tfrac12 (\sbar \Gamma_\mu \psi \Gamma_\nu s) \Gamma^{\mu\nu} s
    + (\sbar \Gamma^\mu s) \psi \Gamma_\mu s\right)
\end{equation}
parametrised by all $s \in S$ and all $\psi \in \Lambda^3 V$.  It is
simply a matter of choosing $s$ and $\psi$ and calculating the resulting
expression using our favourite explicit realisations of the Clifford
algebra to obtain equations for $t_1$ and $t_2$.  We omit the details,
but simply record that the only solution is $t_1=t_2=0$, so that there
is no component of the solution space of equation~\eqref{eq:cc5} of
type $\Lambda^3 V$.

Finally, let $\beta$ be in $2\Lambda^4V$.  As before, the solution
space is an $\fso(V)$-submodule of $\Lambda^4 V \oplus \Lambda^4 V$,
whence it is either all of $2 \Lambda^4 V$, or else given by the image
of
\begin{equation*}
  \Lambda^4 V \ni \phi \mapsto (t_1 \phi, t_2 \phi) \in \Lambda^4 V
  \oplus \Lambda^4 V~,
\end{equation*}
for some $t_1,t_2 \in \RR$ and again we allow for the case of a
zero-di\-men\-sional solution space when $t_1=t_2=0$.  Putting
$\beta_\mu = t_1 \Gamma_\mu \phi + t_2 \phi \Gamma_\mu$, for
$\phi \in \Lambda^4 V$, into equation~\eqref{eq:cc5} we arrive at
the following set of equations for $t_1,t_2\in\RR$:
\begin{equation}
  \label{eq:cc5b4}
  0 = t_1 \left(\tfrac12 (\sbar \Gamma_{\mu\nu} \phi s) \Gamma^{\mu\nu}
    s + (\sbar \Gamma^\mu s) \Gamma_\mu \phi s \right) + t_2
  \left(\tfrac12 (\sbar \Gamma_\mu \phi \Gamma_\nu s) \Gamma^{\mu\nu} s
    + (\sbar \Gamma^\mu s) \phi \Gamma_\mu s\right)
\end{equation}
parametrised by all $s \in S$ and all $\phi \in \Lambda^4 V$.  A
further simplification due to $\fso(V)$-equivariance is the
following.  Since the equations are homogeneous in $s$, we need only
take $s$ from a set consisting of a representative of each
projectivised orbit of $\Spin(V)$ on $S \setminus \left\{0\right\}$.
As shown, for example, in \cite{Bryant-ricciflat}, there are two such
projectivised orbits, distinguished by the causal character of the
associated Dirac current: either null or timelike.  Therefore we need
only consider two such $s$: one in each type of orbit.  Again we omit
the actual details of this calculation and simply record the results:
taking the null orbit, we find that the only equation is $t_2 = -3
t_1$, and imposing this, the equation from the timelike orbit is
automatically satisfied.

In summary, the solution of equation~\eqref{eq:cc5} is
\begin{equation*}
  \beta_\mu = \Gamma_\mu \varphi - 3 \varphi \Gamma_\mu~,
\end{equation*}
for some $\varphi \in \Lambda^4 V$, with the expression for $\gamma$
then following from equation~\eqref{eq:cc1}.
\end{proof}

It should be remarked that equations~\eqref{eq:cc5b3} and
\eqref{eq:cc5b4} can also be solved without recourse to an explicit
matrix realisation of the Clifford algebra by repeated use of the
Fierz identity \eqref{eq:Fierz}.

We have thus computed the cohomology groups $H^{d,2}(\fm,\fg)$ for all
$d\geq 2$ even.  If $d\geq 4$ they are all trivial by
Lemma~\ref{lem:p4} whereas $H^{2,2}(\fm,\fg)\cong \Lambda^4 V$ by
Proposition~\ref{prop:beta}. To prove our first main result on
infinitesimal deformations of $\fg$ and its $\mathbb Z$-graded
subalgebras $\fh$ we also need to determine the analogous groups
$H^{d,2}(\fm,\fh)$ for $\fh$, for all $d\geq 2$ even.

\subsection{Infinitesimal deformations of subalgebras \texorpdfstring{$\fh\subset\fg$}{h < g}}
\label{sec:infin-deform-h}

Let $\fh=\fh_{-2}\oplus\fh_{-1}\oplus\fh_0$ be a $\mathbb Z$-graded
subalgebra of the Poincaré superalgebra $\fg$ which differs only in
zero degree, that is $\fh_{0}\subset\fg_0$ and $\fh_{j}=\fg_{j}$ for
all $j<0$. In this section we first calculate the cohomology
\begin{equation*}
  H^{d,2}(\fm,\fh) = \frac{\ker\partial: C^{d,2}(\fm,\fh) \to
    C^{d,3}(\fm,\fh)}{\partial C^{d,1}(\fm,\fh)}
\end{equation*}
 for all even $d>0$ and then prove Theorem~\ref{thm:first} on the
 filtered deformations of $\fh$. The even $p$-cochains of small degree
 associated to $\fh$  are displayed in
 Table~\ref{tab:even-cochains-small-h} below.

\begin{table}[h!]
  \centering
  \caption{Even $p$-cochains of small degree of $\fh\subset\fg$}
  \label{tab:even-cochains-small-h}
  \begin{tabular}{c*{5}{|>{$}c<{$}}}
    \multicolumn{1}{c|}{} & \multicolumn{5}{c}{$p$} \\\hline
    deg & 0 & 1 & 2 & 3 & 4 \\\hline
    0 & \fh_0 & \begin{tabular}{@{}>{$}c<{$}@{}} S \to S\\ V \to
                    V\end{tabular} & \odot^2 S \to V & & \\\hline
    2 & & V \to \fh_0 & \begin{tabular}{@{}>{$}c<{$}@{}}
                            \Lambda^2 V \to V\\ V \otimes S \to S \\
                            \odot^2 S \to \fh_0 \end{tabular}
         & \begin{tabular}{@{}>{$}c<{$}@{}} \odot^3 S \to S\\
             \odot^2 S \otimes V \to V\end{tabular} & \odot^4 S \to V
        \\\hline
    4 & & & \Lambda^2 V \to \fh_0 & \begin{tabular}{@{}>{$}c<{$}@{}}
           \odot^2 S \otimes V \to \fh_0 \\
           \Lambda^2 V \otimes S \to S \\
           \Lambda^3 V \to V \end{tabular}
             & \begin{tabular}{@{}>{$}c<{$}@{}}
             \odot^4 S \to \fh_0 \\
            \odot^3 S \otimes V \to S \\
            \odot^2 S\otimes \Lambda^2 V\to V \end{tabular} \\\hline
  \end{tabular}
\end{table}

\begin{proposition}
  \label{prop:cohgroups}
  The group $H^{d,2}(\fm,\fh)=0$ for all even $d\geq 4$, whereas
  \begin{equation*}
    H^{2,2}(\fm,\fh) = \frac{\left\{\beta^\varphi + \gamma^\varphi
        + \partial\tilde\alpha\,|\,\varphi\in\Lambda^4
        V,\,\tilde\alpha: V \to \fso(V) \text{~with~}
        \gamma^\varphi(s,s)-\tilde\alpha([s,s])\in
        \fh_0\right\}}{\left\{\partial\tilde\alpha \middle |
        \tilde\alpha: V \to \fh_0\right\}},
  \end{equation*}
  where $(\beta^\varphi,\gamma^\varphi)$ are as in
  Proposition~\ref{prop:beta}.  In particular any cohomology class
  $[\beta^\varphi + \gamma^\varphi + \partial\tilde\alpha] \in
  H^{2,2}(\fm,\fh)$ with $\varphi=0$ is the trivial cohomology class.
\end{proposition}

\begin{proof}
  The proof of the first claim is as in Lemma~\ref{lem:p4} and therefore we omit it.

  It follows from Lemma~\ref{lem:iso} that given any $\alpha \in
  \Hom(\Lambda^2V,V)$, there is a unique $\tilde\alpha \in
  C^{2,1}(\fm,\fg)= \Hom(V,\fso(V))$ such that $\partial \tilde\alpha
  = \alpha + \tilde\beta + \tilde\gamma$, for some $\tilde\beta \in
  \Hom(V \otimes S, S)$ and $\tilde\gamma \in
  \Hom(\odot^2S,\fso(V))$.  Therefore any cochain $\alpha + \beta +
  \gamma\in C^{2,2}(\fm,\fh)$ may be uniquely written as
  \begin{equation*}
    \alpha + \beta + \gamma = (\alpha + \beta + \gamma
    - \partial\tilde\alpha) + \partial\tilde\alpha = (\beta
    -\tilde\beta) +  (\gamma - \tilde\gamma) + \partial \tilde\alpha~,
  \end{equation*}
  where $\beta - \tilde\beta \in \Hom(V\otimes S, S)$ and
  $\gamma-\tilde\gamma \in \Hom(\odot^2S, \fso(V))$.  If $\alpha +
  \beta + \gamma$ is a cocycle, then so is $(\beta - \tilde \beta) +
  (\gamma - \tilde\gamma)$, so that by Proposition~\ref{prop:beta},
  $\beta - \tilde \beta = \beta^\varphi$ and $\gamma - \tilde \gamma =
  \gamma^\varphi$, for some $\varphi \in \Lambda^4 V$ and where
  $\beta^\varphi$ and $\gamma^\varphi$ are given by the expressions in
  Proposition~\ref{prop:beta}.  In other words,
  \begin{equation}
    \label{eq:firstinclusion}
    \ker\partial\bigr|_{C^{2,2}(\fm,\fh)} \subset \Lambda^4
    V\oplus \partial(\fso(V)\otimes V^*)~,
  \end{equation}
  where we identified any $\varphi\in\Lambda^4 V$ with the corresponding
  cocycle $\beta^\varphi+\gamma^\varphi$. 
	Now equation \eqref{eq:Spencer1} tells
  us that
  \begin{equation*}
    \partial\tilde\alpha(s,s)=-\tilde\alpha([s,s])  
  \end{equation*}
  for all $s\in S$ so that an element $\beta^\varphi + \gamma^\varphi
  + \partial\tilde\alpha$ is in $C^{2,2}(\fm,\fh)$ if and only if
  \begin{equation*}
    \gamma^\varphi(s,s) - \tilde\alpha([s,s])\in \fh_0
  \end{equation*}
  for all $s\in S$.  This fact, together with
  \eqref{eq:firstinclusion}, shows that the kernel of $\partial$
  restricted to $C^{2,2}(\fm,\fh)$ is given by
  \begin{equation*}
   \left\{\beta^\varphi+\gamma^\varphi+\partial\tilde\alpha\middle |
     \varphi\in\Lambda^4 V,\,\tilde\alpha: V \to \fso(V) \text{~with~}
     \gamma^\varphi(s,s)-\tilde\alpha([s,s])\in
     \fh_0^{\phantom{C}\!}\right\}~, 
  \end{equation*}
  from which the claim on $H^{2,2}(\fm,\fh)$ follows directly.

  The last claim follows from the fact that, if $\varphi=0$, then
  $\partial\tilde\alpha$ satisfies $\tilde\alpha([s,s])\in\fh_0$ for
  all $s\in S$, so that it is in the image of
  $C^{2,1}(\fm,\fh)=\Hom(V,\fh_0)$.
\end{proof}
	
To state our first main result on filtered deformations $F$ of $\fh$
we recall that the Lie brackets of $F$ have components of nonzero
filtration degree: the component $\mu$ of degree $2$ (see equation
\eqref{eq:filteredef}) and the component $\delta:\Lambda^2 V\to \fh_0$
of degree $4$.

\begin{theorem}
  \label{thm:first}
  Let $\fh=\fh_{-2}\oplus\fh_{-1}\oplus\fh_0$ be a $\mathbb Z$-graded
  subalgebra of the Poincaré superalgebra $\fg=V\oplus S\oplus
  \fso(V)$ which differs only in zero degree, i.e.
  $\fh_{0}\subset\fg_0$ and $\fm=\fh_{-2}\oplus\fh_{-1}=V\oplus S$. If
  $F$ is a filtered deformations of $\fh$ then:
  \begin{enumerate}
  \item $\mu|_{\fm\otimes\fm}$ is a cocycle in $C^{2,2}(\fm,\fh)$ and
    its cohomology class $[\mu|_{\fm\otimes\fm}]\in H^{2,2}(\fm,\fh)$
    is $\fh_0$-invariant (that is $\mu|_{\fm\otimes\fm}$ is
    $\fh_0$-invariant up to exact terms); and
  \item if $F'$ is another filtered deformation of $\fh$ such that
    $[\mu'|_{\fm\otimes\fm}]= [\mu|_{\fm\otimes\fm}]$ then $F'$ is
    isomorphic to $F$ as a filtered Lie superalgebra.
  \end{enumerate}
\end{theorem}

\begin{proof}
  By the results of \cite{MR3255456}, the maximal transitive
  prolongation of the supertranslation algebra $\fm$ is the
  $\ZZ$-graded Lie superalgebra
  $\fg^{\infty}=\fg^{\infty}_{-2}\oplus\fg^{\infty}_{-1}\oplus\fg^{\infty}_{0}$
  where $\fg^{\infty}_{-2}\oplus\fg^{\infty}_{-1}=\fm$ and
  $\fg_0^{\infty}=\fso(V)\oplus \mathbb R E$, with $E$ the so-called
  grading element satisfying
  $\ad(E)|_{\fg^{\infty}_{j}}=j\operatorname{Id}$. It is well-known
  that maximality is equivalent to the fact that the Spencer
  cohomology group $H^{d,1}(\fm,\fg^{\infty})=0$ for all $d\geq 0$
  (see e.g. \cite{MR1688484}).
  
  Since $\fm=\fh_{-2}\oplus\fh_{-1}$ but
  $\fh_0\subset\fso(V)\subset\fg^{\infty}_0$ this also implies that
  $H^{d,1}(\fm,\fh)=0$ for all $d\geq 1$, that is $\fh$ is a
  \emph{full prolongation of degree k=1} in the terminology of
  \cite{MR1688484}.
  
  The first claim follows directly from Proposition~2.2 of
  \cite{MR1688484}. Let now $F$ and $F'$ be two filtered deformations
  of $\fh$ such that
  $[\mu|_{\fm\otimes\fm}]=[\mu'|_{\fm\otimes\fm}]$. Then
  $(\mu-\mu')|_{\fm\otimes\fm}$ is a Spencer coboundary and we may
  first assume without any loss of generality that
  $\mu|_{\fm\otimes\fm}=\mu'|_{\fm\otimes\fm}$ by Proposition~2.3 of
  \cite{MR1688484}. Moreover, since $\fh$ is a full prolongation of
  degree $k=1$ (and hence, in particular, of degree $k=2$),
  Proposition~2.6 of \cite{MR1688484} applies and we may also assume
  $\mu=\mu'$ without any loss of generality. In other words we just
  showed that $F'$ is isomorphic as a filtered Lie superalgebra to
  another filtered Lie superalgebra $F''$ which satisfies $\mu''=\mu$.
  
  Now, given any two filtered deformations $F$ and $F'$ of $\fh$ with
  $\mu=\mu'$ it is not difficult to see that
  $\delta-\delta'=(\delta-\delta')|_{\fm\otimes\fm}$ is a Spencer
  cocycle (use e.g. \cite[equation 2.6]{MR1688484}). However
  $H^{4,2}(\fm,\fh)=\operatorname{Ker}\partial|_{C^{4,2}(\fm,\fh)}=0$
  by Proposition~\ref{prop:cohgroups} and hence $\delta=\delta'$. This 
  proves that any two filtered deformations $F$ and $F'$ of $\fh$ with
  $[\mu'|_{\fm\otimes\fm}]= [\mu|_{\fm\otimes\fm}]$ are isomorphic.
\end{proof}

In other words, filtered deformations of $\fh$ are completely
determined by the $\fh_0$-invariant elements in $H^{2,2}(\fm,\fh)$, a
group which we already calculated in Proposition~\ref{prop:cohgroups}.
We emphasise that this result in particular says that the components
$\rho=\mu|_{\fh_0\otimes V}:\fh_0\otimes V\to \fh_0$ and
$\delta:\Lambda^2V\to \fh_0$ of non-zero filtration degree are
completely determined by the class $[\mu|_{\fm\otimes\fm}]\in
H^{2,2}(\fm,\fh)$ (hence by the components $\alpha$, $\beta$ and
$\gamma$), up to isomorphisms of filtered Lie superalgebras.

\section{Integrating the deformations}
\label{sec:integ-deform}

In this section, we will determine the $\fh_0$-invariant elements in
$H^{2,2}(\fm,\fh)$ and, for each of them, construct a filtered
deformation.  Let us remark that we do not have at our disposal a
bracket à la Nijenhuis--Richardson on $H^{\bullet,\bullet}(\fm,\fh)$
that allows one to write down the obstructions to integrating an
infinitesimal deformation in terms of classes in $H^{\bullet,3}(\fm,\fh)$.
Therefore our description of filtered Lie superalgebras will be very
explicit and rely on a direct check of the Jacobi identities.
	
\subsection{The non-trivial deformations}
\label{sec:defnontrivial}

By the results of Section \ref{sec:infin-deform}, we need only
consider deformations corresponding to $\fh_0$-invariant cohomology
classes in $H^{2,2}(\fm,\fh)$ with $\varphi\neq 0$. Indeed if
$\varphi=0$ then $[\mu|_{\fm\otimes\fm}]=0$ by
Proposition~\ref{prop:cohgroups} and Theorem~\ref{thm:first}, and the
associated Lie superalgebras are nothing but the $\mathbb Z$-graded
subalgebras of the Poincaré superalgebra.

In determining the $\fh_0$-invariant classes in $H^{2,2}(\fm,\fh)$, we
will also determine the Lie subalgebras $\fh_0 \subset \fso(V)$ for
which $H^{2,2}(\fm,\fh)^{\fh_0} \neq 0$ and hence the graded Lie
subalgebras $\fh$ of the Poincaré superalgebra admitting nontrivial
filtered deformations.  We will show that the condition
$H^{2,2}(\fm,\fh)^{\fh_0} \neq 0$ turns into a system of quadratic
equations for $\varphi$ and $\tilde\alpha$ which we will be able to
solve.  In addition we will find that $\fh_0 = \fh_\varphi$, the
Lie algebra of the stabiliser in $\SO(V)$ of $\varphi$; that is, 
$\fh_\varphi = \fso(V) \cap \stab(\varphi)$, with $\stab(\varphi)$ the
Lie algebra of the stabiliser of $\varphi$ in $\GL(V)$.  We start with
a lemma.

\begin{lemma}
  \label{lem:stab}
  Let $\beta^\varphi+\gamma^\varphi+\partial\tilde\alpha$ be a cocycle
  in $C^{2,2}(\fm,\fh)$ defining a nontrivial, $\fh_0$-invariant
  cohomology class in $H^{2,2}(\fm,\fh)$.  Then $\fh_0$ leaves
  $\varphi$ invariant.  In other words, $\fh_0 \subset \fh_\varphi$.
\end{lemma}

\begin{proof}
  Let $\beta^\varphi+\gamma^\varphi+\partial\tilde\alpha$ be a cocycle
  in $C^{2,2}(\fm,\fh)$ such that its cohomology class in
  $H^{2,2}(\fm,\fh)$ is non-trivial and $\fh_0$-invariant. For our
  purposes, it is convenient to consider the decomposition of
  $\fso(V)$-modules
  \begin{equation*}
    C^{2,2}(\fm,\fg)=\Hom(\Lambda^2 V, V)\oplus \Hom(V \otimes S, S)
    \oplus \Hom(\odot^2S, \fso(V))
  \end{equation*}
  and the corresponding $\fso(V)$-equivariant projections
  \eqref{eq:projectors}.  We recall that $\tilde\alpha:V\to \fso(V)$
  is such that
  $\gamma^\varphi+\pi^\gamma(\partial\tilde\alpha):\odot^2 S\to
  \fh_0$. Now $\varphi$ is nonzero, by
  Proposition~\ref{prop:cohgroups}, and for any $x\in\fh_0$ there is a
  $\psi\in C^{2,1}(\fm,\fh)=\Hom(V,\fh_0)$ such that
  $x\cdot(\beta^\varphi+\gamma^\varphi+\partial\tilde\alpha)=\partial\psi$.
  In other words,
  \begin{align}
    \label{eq:a}
    x\cdot(\pi^{\alpha}(\partial\tilde\alpha))&=\pi^{\alpha}(\partial\psi)~,\\
    \label{eq:b}
    x\cdot(\beta^\varphi+\pi^\beta(\partial\tilde\alpha))&=\pi^{\beta}(\partial\psi)~,\\
    \label{eq:c}
    x\cdot(\gamma^\varphi+\pi^{\gamma}(\partial\tilde\alpha))&=\pi^{\gamma}(\partial\psi)~.
  \end{align}
  From equation~\eqref{eq:a} and the $\fso(V)$-equivariance of
  $\pi^\alpha$ and $\partial$, we have
  \begin{equation*}
    (\pi^{\alpha}\circ\partial)(\psi) =
    x\cdot(\pi^{\alpha}(\partial\tilde\alpha)) =
    (\pi^{\alpha}\circ\partial) (x\cdot\tilde\alpha)
  \end{equation*}
  and then, since $\pi^{\alpha}\circ\partial:\Hom(V,\fso(V)) \to
  \Hom(\Lambda^2 V,V)$ is an isomorphism by Lemma~\ref{lem:iso}, it
  follows that $x\cdot\tilde\alpha=\psi$.  Equation \eqref{eq:b} yields now
  \begin{equation*}
    \begin{split}
      \pi^{\beta}(\partial\psi) &= x\cdot \left(\beta^\varphi +
        \pi^{\beta} (\partial\tilde\alpha)\right)\\
      &= x \cdot \beta^\varphi + x \cdot \pi^{\beta}(\partial\tilde\alpha) \\
      &= x\cdot\beta^\varphi+\pi^{\beta}(\partial(x\cdot
      \tilde\alpha))=x\cdot\beta^\varphi+\pi^{\beta}(\partial\psi)
    \end{split}
  \end{equation*}
  so $x\cdot\beta^\varphi=0$ and, by a similar argument starting with
  equation \eqref{eq:c}, $x\cdot\gamma^\varphi=0$ too. This shows that
  $\varphi$ is invariant by $\fh_0$ or, in other words, that
  $\fh_0\subset \fh_\varphi$.
\end{proof}

It follows from this lemma, that if the cocycle $\beta^\varphi +
\gamma^\varphi + \partial\tilde\alpha$ defines a nontrivial,
$\fh_0$-invariant cohomology class in $H^{2,2}(\fm,\fh)$, then in
particular the component
$\gamma^\varphi+\pi^\gamma(\partial\tilde\alpha)$ in $\Hom(\odot^2S,
\fh_0)$ actually belongs to $\Hom(\odot^2 S, \fh_\varphi)$ and this
has some strong consequences.  To exhibit them, we need a technical
lemma.

\begin{lemma}
  \label{lem:quadrics}
  The cochain $\gamma^\varphi+\pi^\gamma(\partial\tilde\alpha)$ takes
  values in $\fh_\varphi$ if and only if $\varphi \in \Lambda^4 V$ and
  $\tilde\alpha \in \Hom(V,\fso(V))$ satisfy the following three
  systems of quadrics:
  \begin{equation}
    \label{eq:quadrics}
    \begin{aligned}[m]
      \operatorname{skew}_{\lambda_1,\dots,\lambda_6} (\varphi^{\lambda_1\lambda_2\lambda_3\lambda_4}
      \varphi^{\lambda_5[\mu_1\mu_2\mu_3} \eta^{\mu_4]\lambda_6}) &=
      0~,\\
      \varphi_{\rho\mu\nu}{}^{[\mu_1} \varphi^{\mu_2\mu_3\mu_4]\rho}
      &= 0~,\\
      \tilde\alpha_{\lambda\rho}{}^{[\mu_1}
      \varphi^{\mu_2\mu_3\mu_4]\rho} &= 0~,
    \end{aligned}
  \end{equation}
  where in all three formulae we skew-symmetrise in the $\mu_i$ and,
  in addition, in the first formula we skew-symmetrise in the
  $\lambda_i$ as well, but separately.
\end{lemma}

\begin{proof}
  The cochain $\gamma^\varphi + \pi^\gamma(\partial\tilde\alpha)$
  takes values in $\fh_\varphi$ if and only if for every $s \in S$,
  $\gamma(s,s) -\tilde\alpha([s,s]) \in \fso(V)$ leaves $\varphi$
  invariant, where $[s,s]$ stands for the Dirac current of $s$.
  Relative to an $\eta$-orthonormal basis, $\varphi = \tfrac1{4!}
  \varphi^{\mu_1 \cdots \mu_4} \be_{\mu_1} \wedge \cdots \wedge
  \be_{\mu_4}$ and
  \begin{equation*}
    \gamma(s,s)(\varphi) = \tfrac1{3!} \varphi^{\mu_1\cdots\mu_4}
    \gamma(s,s)^\nu{}_{\mu_1} \be_\nu \wedge \be_{\mu_2} \wedge \cdots
    \wedge \be_{\mu_4}~.
  \end{equation*}
  Using equation~\eqref{eq:gammabasis}, this becomes
  \begin{equation}
    \label{eq:Finvt1}
    \gamma(s,s)(\varphi) = \tfrac1{3} \varphi^{\mu_1\cdots\mu_4}
    (\sbar \Gamma^\nu \beta_{\mu_1} s) \be_\nu \wedge \be_{\mu_2} \wedge
    \cdots \wedge \be_{\mu_4}~,
  \end{equation}
  where, from Proposition~\ref{prop:beta},
  \begin{equation*}
    \beta_\rho = \tfrac1{4!} \varphi^{\lambda_1\cdots \lambda_4}
    \left(\Gamma_\rho \Gamma_{\lambda_1\cdots\lambda_4} - 3
      \Gamma_{\lambda_1\cdots\lambda_4} \Gamma_\rho \right)~.
  \end{equation*}
  On the other hand,
  \begin{equation*}
   - \tilde\alpha([s,s])(\varphi) = \tfrac1{3!}
   \varphi^{\mu_1\cdots\mu_4} (\sbar\Gamma^\rho s)
   \tilde\alpha_\rho{}^\nu{}_{\mu_1} \be_\nu \wedge \be_{\mu_2} \wedge
    \cdots \wedge \be_{\mu_4}~.
  \end{equation*}
  We insert the expression for $\beta_{\mu_1}$ into
  equation~\eqref{eq:Finvt1}, multiply in $\Cl(V)$ and keep only the
  terms $\sbar \Gamma_{\mu_1\cdots\mu_p} s$ for $p=1,2,5$ (or,
  equivalently, $6$).  When the dust clears, we find that
  $(\gamma(s,s) -\tilde\alpha([s,s]))(\varphi) = 0$ if and only
  if
  \begin{equation*}
    \tfrac1{6} \varphi^{\lambda_1\cdots\lambda_4}
    \varphi^{\lambda_5[\mu_1\cdots\mu_3} \eta^{\mu_4]\lambda_6}
    \sbar \Gamma_{\lambda_1\cdots\lambda_6}s + 8
    \varphi_{\rho\mu\nu}{}^{[\mu_1}\varphi^{\mu_2\cdots\mu_4]\rho}
    \sbar\Gamma^{\mu\nu} s - \tilde\alpha_{\lambda\rho}{}^{[\mu_1}
    \varphi^{\mu_2\cdots\mu_4]\rho} \sbar \Gamma^\lambda s = 0~.
  \end{equation*}
  By polarisation on $s$, we see that this is a system of quadrics for
  $\varphi$ and $\tilde\alpha$ with linear parametric dependence on
  $\odot^2 S$.  Since $\odot^2 S \cong V \oplus \Lambda^2 V \oplus
  \Lambda^5 V$, the components of this system parametrised by $V$,
  $\Lambda^2 V$ and $\Lambda^5 V \cong \Lambda^6 V$ must be satisfied
  separately.  In other words, the terms proportional to
  $\sbar\Gamma^\lambda s$, $\sbar\Gamma^{\mu\nu} s$ and 
  $\sbar \Gamma_{\lambda_1\cdots\lambda_6}s$ must vanish separately, and
  these are precisely the three systems of quadrics in the lemma.
\end{proof}

As we shall see, the quadrics \eqref{eq:quadrics} have a very natural
interpretation.  Our first observation is that the first equation in
\eqref{eq:quadrics} actually implies the second.  To see this, we
simply contract the first equation with $\eta_{\lambda_6\mu_4}$ to
obtain
\begin{equation}
  \label{eq:plueckerinds}
  \varphi^{[\lambda_1\lambda_2\lambda_3\lambda_4}
  \varphi^{\lambda_5]\mu_1\mu_2\mu_3} = 0~,
\end{equation}
and we now contract again with $\eta_{\lambda_5\mu_3}$ to obtain the
second equation.

Now recall that a non-zero $4$-form $\varphi \in \Lambda^4 V$
is said to be \emph{decomposable} if
\begin{equation}
  \label{eq:decomposable}
  \varphi = v_1 \wedge v_2 \wedge v_3 \wedge v_4~,
\end{equation}
for some linearly independent $v_i \in V$.  If
$\varphi = v_1 \wedge v_2 \wedge v_3 \wedge v_4$ is decomposable, then
the first equation in \eqref{eq:quadrics} (and hence also the second)
is satisfied identically.  To see this, insert
$\varphi^{\lambda_1\cdots\lambda_4} = v_1^{[\lambda_1} v_2^{\lambda_2}
v_3^{\lambda_3} v_4^{\lambda_4]}$,
into the LHS of the first equation in \eqref{eq:quadrics} to obtain
\begin{equation*}
  \operatorname{skew}_{\mu_1,\cdots,\mu_4}
  \operatorname{skew}_{\lambda_1,\cdots,\lambda_6}
  \left(v_1^{[\lambda_1} v_2^{\lambda_2} v_3^{\lambda_3} v_4^{\lambda_4]}
  v_1^{[\lambda_5}v_2^{\mu_1} v_3^{\mu_2} v_4^{\mu_3]}
  \eta^{\mu_4\lambda_6}\right)~,
\end{equation*}
where we skew-symmetrise separately in the $\lambda_i$ and the
$\mu_i$.  But notice that every term in this expression contains a
factor $v_i^{\lambda_j} v_i^{\lambda_k}$ for some $i,j,k$, and this
vanishes by symmetry since we skew-symmetrise on the $\lambda_i$.
	
Perhaps more remarkable still is that the converse also holds. Indeed,
we recognise equation~\eqref{eq:plueckerinds} as the Plücker relations
(see, e.g., \cite[Ch.~1]{MR507725})
\begin{equation}
  \label{eq:plueckerGH}
  \iota_\chi \iota_\theta\iota_\zeta \varphi \wedge \varphi = 0~,
\end{equation}
for all $\theta,\zeta,\chi \in V^*$, defining the Plücker embedding of
the grassmannian $\Gr(4,V)$ of $4$-planes in $V$ into the projective
space $\PP(\Lambda^4 V)$.  Recall that a decomposable 4-form $\varphi
= v_1 \wedge v_2 \wedge v_3 \wedge v_4$ defines a plane $\Pi \subset
V$ by the span of the $(v_i)$ and, conversely, any plane determines a
decomposable $\varphi$ up to a nonzero real multiple by taking the
4-form constructed out of wedging the elements in any basis.  Hence
$\varphi$ is decomposable if and only if it obeys
equation~\eqref{eq:plueckerGH}.  In other words, we have proved that
the first two equations in \eqref{eq:quadrics} are satisfied if and
only if $\varphi$ is decomposable.

Finally, the third quadric in \eqref{eq:quadrics} simply says that the
image of $\tilde\alpha : V \to \fso(V)$ actually lies in $\fh_\varphi$.

In summary, we have proved the following

\begin{proposition}
 \label{prop:pluecker}
  The cochain $\gamma^\varphi+\pi^\gamma(\partial\tilde\alpha)$ takes
  values in $\fh_\varphi$ if and only if $\varphi \in \Lambda^4 V$ is
  decomposable and the image $\Im(\tilde\alpha)\subset\fh_\varphi$.
\end{proposition}

To proceed further, we need to classify the decomposable $\varphi$ and
the corresponding stabilisers $\fh_\varphi$.  It is only necessary to
classify $\varphi$ up to the action of $\CSO(V) = \RR^\times \times
\SO(V)$.

\begin{lemma}
  \label{lem:varphi}
  Let $\varphi$ and $\varphi'$ be decomposable $4$-forms in the same
  orbit of $\CSO(V)$ on $\Lambda^4 V$.  Then the corresponding
  filtered deformations are isomorphic.
\end{lemma}

\begin{proof}
  The group $\mathrm{G}^{\infty}_0=\mathrm{CSpin(V)}$ with Lie algebra
  $\fg_0^{\infty}=\mathfrak{so}(V)\oplus\mathbb R E$ is a double-cover
  of $\CSO(V)$ and it naturally acts  on $\mathfrak{g}^\infty$ by
  degree-$0$ Lie superalgebra automorphisms. Note that the action
  preserves the Poincaré superalgebra $\fg$, which is an ideal of
  $\fg^{\infty}$. Now, an element $g\in\mathrm{CSpin(V)}$ sends a
  $\mathbb Z$-graded subalgebra $\fh=V\oplus S\oplus \fh_0$ of $\fg$
  into an (isomorphic) $\mathbb Z$-graded subalgebra
  $\fh'=g\cdot\fh=V\oplus S\oplus (g\cdot \fh_0)$ of $\fg$. In
  particular, if $F$ is a filtered deformation of $\fh$ associated with
  $\varphi$ then $F'=g\cdot F$ is also a filtered deformation of $\fh'$,
  which is associated with $\varphi'=g\cdot\varphi$.
\end{proof}

Therefore we must classify the orbits of $\CSO(V)$ in the space of
decomposable elements of $\Lambda^4 V$.  Other than $\varphi = 0$,
which is its own orbit, any other decomposable $\varphi$ defines a
$4$-plane and we can study instead the geometric action of $\SO(V)$ on
$4$-planes.  Unlike the general linear group, $\SO(V)$ does not act
transitively on the grassmannian of $4$-planes.  Indeed, we can
distinguish three kinds of planes, depending on the nature of the
restriction of the inner product $\eta$ on $V$ to the plane:
\begin{enumerate}
\item $\Pi$ is euclidean: we will say that $\varphi$ is \emph{spacelike};
\item $\Pi$ is lorentzian: we will say that $\varphi$ is \emph{timelike};
\item $\Pi$ is degenerate: we will say that $\varphi$ is \emph{lightlike}.
\end{enumerate}
Since $\SO(V)$ preserves $\eta$, it preserves the type of plane and
acts transitively on each type.   In terms of the 4-forms, one can
show that, in addition to the trivial orbit $\varphi=0$, there are
precisely three orbits of $\CSO(V)$ on the space of decomposable
elements in $\Lambda^4 V$.

Many of the results we prove from here on depend on a case-by-case
analysis of these three orbits.  We find that the first two orbits
can be treated simultaneously, since they share the property that the
restriction of $\eta$ to $\Pi$ is nondegenerate.  In this case, we can
decompose $V = \Pi \oplus \Pi^\perp$ into an orthogonal direct sum and
hence $\fh_\varphi = \fso(\Pi) \oplus \fso(\Pi^\perp) \subset
\fso(V)$.

In contrast, if $\Pi$ is degenerate, we can always choose an
$\eta$-Witt basis for $V$ such that $V = \RR\left<\be_+,\be_-\right>
\oplus W$ and such that $\varphi = \be_+ \wedge f$ for $f \in
\Lambda^3 W$ a decomposable $3$-form.  Such $f$ defines a $3$-plane
$\pi \subset W$ and induces an orthogonal decomposition $W = \pi
\oplus \pi^\perp$.  Our original plane is $\Pi = \RR\left<\be_+\right>
\oplus \pi$ and the stabiliser Lie algebra is now
\begin{equation*}
  \fh_\varphi = \left(\fso(\pi) \oplus \fso(\pi^\perp)\right) \ltimes
  (\be_+ \wedge (\pi \oplus \pi^\perp)) \subset \fso(V)~,
\end{equation*}
where
\begin{equation*}
  \be_+ \wedge (\pi \oplus \pi^\perp) = (\be_+\wedge \pi) \oplus
  (\be_+ \wedge \pi^\perp)~,
\end{equation*}
is the abelian Lie subalgebra of $\fso(V)$ consisting of null
rotations fixing $\be_+$.  We remark that whether or not $\Pi$ is
degenerate, $\dim\fh_\varphi = 27$ and in fact the degenerate
$\fh_\varphi$ is a contraction of the nondegenerate $\fh_\varphi$.

We are now ready to prove the following proposition, which
recapitulates the results of this section.

\begin{proposition}
  \label{thm:second}
  Let $\fh=\fh_{-2}\oplus\fh_{-1}\oplus\fh_0$ be a $\mathbb Z$-graded
  subalgebra of the Poincaré superalgebra $\fg$ which differs only in
  zero degree; that is, $\fh_{0}\subset\fg_0$ and
  $\fm=\fh_{-2}\oplus\fh_{-1}=V\oplus S$.  In addition, let
  $\beta^\varphi+\gamma^\varphi+\partial\tilde\alpha$ be a cocycle in
  $C^{2,2}(\fm,\fh)$ defining a nontrivial, $\fh_0$-invariant
  cohomology class in $H^{2,2}(\fm,\fh)$.  Then,
  \begin{enumerate}
  \item  $\varphi\in\Lambda^4 V$ is nonzero and decomposable,
  \item the images $\Im(\gamma^\varphi)=\fh_\varphi$ and $\Im(\tilde\alpha)
    \subset \fh_\varphi$, and
  \item $\fh_0=\fh_\varphi$.
  \end{enumerate}
\end{proposition}

\begin{proof}
  The first part follows from Lemma~\ref{lem:stab} and
  Proposition~\ref{prop:pluecker}.  Further, from Lemma~\ref{lem:stab}
  we have that for all $s \in S$,
  \begin{equation*}
    \gamma^\varphi(s,s) - \tilde\alpha([s,s]) \in \fh_\varphi~,
  \end{equation*}
  but from Proposition~\ref{prop:pluecker} we also know that
  $\tilde\alpha([s,s]) \in \fh_\varphi$, hence $\Im(\gamma^\varphi)
  \subset \fh_\varphi$ as well.

  To prove the second part, we need to show that $\Im(\gamma^\varphi)
  = \fh_\varphi$.  We break this up into two cases, depending on
  whether or not the plane $\Pi$ corresponding to $\varphi$ is
  degenerate.

  \emph{$\Pi$ is nondegenerate.}  From Proposition~\ref{prop:beta}, we
  have that
  \begin{equation*}
    \eta(w,\gamma^\varphi(s,s) v) = 2 \left<s, w \cdot \beta^\varphi_v
    \cdot s\right>~.
  \end{equation*}
  Writing $v = v_\top + v_\perp$ and $w = w_\top + w_\perp$, and using
  that $\beta^\varphi_v = 4 v_\top\cdot \varphi - 2 v_\perp\cdot \varphi$ we
  arrive after some calculation at
  \begin{equation*}
    \eta(w,\gamma^\varphi(s,s) v) = 8 \left<s, \iota_{w_\top}
      \iota_{v_\top} \varphi \cdot s\right> - 4 \left<s, w_\perp
      \wedge v_\perp \wedge \varphi \cdot s\right>~.
  \end{equation*}
  In other words, $\gamma^\varphi$ defines a map
  \begin{equation*}
    \odot^2 S \to \fso(\Pi) \oplus \fso(\Pi^\perp)
  \end{equation*}
  which we claim is surjective.  Indeed, the only way that the first
  component of this map fails to be surjective is if there exists a
  nonzero $\zeta \in \Lambda^2 \Pi$ such that
  \begin{equation*}
    \left<s, \iota_\zeta \varphi \cdot s\right>= 0 \quad \forall s \in S~.
  \end{equation*}
  From equation~\eqref{eq:nondeg2forms}, this is true if and only if
  $\iota_\zeta \varphi = 0$, but this implies that $\zeta = 0$.
  Similarly, the second component of the map would fail to be surjective
  if there exists a nonzero $\theta \in \Lambda^2 \Pi^\perp$ such that
  \begin{equation*}
    \left<s, \theta \wedge \varphi \cdot s\right>= 0 \quad \forall s
    \in S~.
  \end{equation*}
  In turn, this is equivalent to
  \begin{equation*}
    \left<s, \star(\theta \wedge \varphi) \cdot s\right>= 0 \quad \forall s
    \in S~,
  \end{equation*}
  which, by \eqref{eq:nondeg5forms}, implies that $\star(\theta \wedge
  \varphi) = 0$ or, equivalently, $\theta \wedge \varphi = 0$; but no
  nonzero $\theta \in \Lambda^2\Pi^\perp$ has vanishing wedge
  product with $\varphi$.  Therefore $\theta = 0$.

  \emph{$\Pi$ is degenerate.}   In this case we can write $w = w_+ +
  w_- + w_\top + w_\perp$ and now
  \begin{equation*}
    \beta_v^\varphi = 4 v_\top\cdot \varphi - 2 v_\perp\cdot \varphi - 2 v_-\cdot
    \varphi - 6 \theta(v_-) f~,
  \end{equation*}
  where $\varphi = \be_+ \wedge f$ and $\theta(v_-) = \eta(v,\be_+)$.
  After a short calculation, we arrive at
  \begin{equation*}
    \begin{split}
      \eta(w,\gamma^\varphi(s,s) v) &= 2 \left<s, w \cdot
        \beta^\varphi_v \cdot s\right>\\
      &= 8 \theta(w_-) \left<s, \iota_{v_\top} f \cdot s\right> - 8
      \theta(v_-) \left<s, \iota_{w_\top} f \cdot s\right>\\
      & \quad {} - 4 \left<s, w_\perp \wedge v_\perp \wedge \varphi
        \cdot s \right> + 8 \left<s, \iota_{w_\top} \iota_{v_\top} \varphi
      \cdot s\right>\\
      & \quad {} - 4 \left<s, w_- \wedge v_\perp \wedge \varphi \cdot
      s\right> + 4 \left<s, v_- \wedge w_\perp \wedge \varphi \cdot
      s\right> ~.
    \end{split}
  \end{equation*}
  The first two terms factor through the component
  $ \odot^2S \to (\be_+ \wedge \pi)$ of $\gamma^\varphi$, whereas the
  second pair of terms factor through the component
  $\odot^2S \to (\be_+ \wedge \pi^\perp)$.  The last two terms factor
  through the components $\odot^2S \to \fso(\pi)$ and
  $\odot^2S \to \fso(\pi^\perp)$ of $\gamma^\varphi$, respectively.
  Similar arguments to the ones in the nondegenerate case show that
  these maps are surjective.

  Finally, we show that $\fh_0 = \fh_\varphi$.  From
  Lemma~\ref{lem:stab} we know that $\fh_0 \subset \fh_\varphi$, so
  all we need to do is to establish the reverse inclusion:
  $\fh_0\supset \fh_\varphi$.  This will follow from
  \begin{equation*}
    \fh_0 \supset \Im( \gamma^\varphi+\pi^\gamma(\partial
    \tilde\alpha) ) = \Im(\gamma^\varphi) + \Im(\pi^\gamma(\partial
    \tilde\alpha) ) = \fh_\varphi~,
  \end{equation*}
  where the first equality is a consequence of the fact, to be shown,
  that we may actually think of $\gamma^\varphi + \pi^\gamma(\partial
  \tilde\alpha)$ as
  \begin{equation}
    \label{eq:directsummaps}
    \gamma^\varphi \oplus \pi^\gamma(\partial \tilde\alpha) :
    (\Lambda^2 V \oplus \Lambda^5 V) \oplus V \to \fh_\varphi~,
  \end{equation}
  where we identify $\odot^2 S$ with the direct sum $V
  \oplus \Lambda^2 V \oplus \Lambda^5 V$ of subspaces $\Lambda^p V \subset
  \odot^2 S$, for $p=1,2,5$, defined by
  equation~\eqref{eq:lambdaVinSS}.

  It follows from the very definition of
  $\pi^\gamma(\partial\tilde\alpha)$ that it is given by a map
  $V \to \fh_\varphi$: in fact, the map is precisely $-\tilde\alpha$.
  We now use equations~\eqref{eq:VonForms} to calculate
  \begin{equation*}
    \eta(w,\gamma^\varphi(s,s)v) = 4 \left<s, v \wedge w \wedge
      \varphi \cdot s \right> - 8 \left<s, \iota_v \iota_w \varphi
      \cdot s\right>~,
  \end{equation*}
  which shows that $\gamma^\varphi : \Lambda^2 V \oplus \Lambda^5 V
  \to \fh_\varphi$, where we have used that $\Lambda^6 V \cong
  \Lambda^5 V$.
  \end{proof}

\subsection{First-order integrability of the deformation}
\label{sec:first-order-integr-deform}

From Proposition~\ref{prop:cohgroups}, Theorem~\ref{thm:first},
Lemma~\ref{lem:varphi} and Proposition~\ref{thm:second} we know that
there are (at most) three isomorphism classes of non-trivial filtered
deformations of subalgebras $\fh=V\oplus S\oplus\fh_0$ of the Poincaré
superalgebra, each one determined by a nonzero decomposable
$\varphi\in\Lambda^4V$ which can be either spacelike, timelike or
lightlike.  Moreover $\fh_0 = \fh_{\varphi} = \fso(V)\cap\stab(\varphi)$ and
\begin{equation*}
  H^{2,2}(\fm,\fh)^{\fh_0} =
  \frac{\left\{\beta^\varphi+\gamma^\varphi+\partial\tilde\alpha
      \middle | \tilde\alpha:V\to
      \fh_0\right\}}{\left\{\partial\tilde\alpha\,|\,\tilde\alpha:V\to
      \fh_0\right\}} \cong \left\{\beta^\varphi+\gamma^\varphi \middle
  | \varphi \in \left(\Lambda^4V\right)^{\fh_0}\right\}~.
\end{equation*}
In other words the action of $\mu$ on $V\otimes S$ and $\odot ^2S$ is
given by $\beta^\varphi$ and $\gamma^\varphi$ (recall that
$\gamma^\varphi:\odot^2 S\to\fso(V)$ already takes values in
$\fh_0$ when $\varphi$ is decomposable) and one can always assume
$\alpha=\mu|_{\Lambda^2 V}=0$ without loss of generality.

Let us now introduce a formal parameter $t$ to keep track of the order of
the deformation. The original graded Lie superalgebra structure has
order $t^0$ and the infinitesimal deformation has order $t$. We will
now show that, to first order in $t$, the filtered Lie superalgebra
structure on $\fh$ is given by
\begin{equation*}
  \begin{aligned}[m]
    [v_1,v_2] &= 0\\
    [v,s] &= t \beta^\varphi_v(s) = t (v \cdot \varphi - 3 \varphi \cdot v) \cdot s\\
    [s_1,s_2] &= [s_1,s_2]_0 + t \gamma^\varphi(s_1,s_2)~,
  \end{aligned}
\end{equation*}
where $[s_1,s_2]_0$ denotes the original Lie bracket defined by
equation~\eqref{eq:DiracCurrent} and the brackets involving
$\fh_0$ are unchanged. In particular we set $\rho=\mu|_{\fh_0\otimes
  V}:\fh_0\otimes V\to \fh_0$ to be zero.

We now check that all the Jacobi identities are satisfied to first order in $t$.
For example, the identity
\begin{equation*}
  [\lambda, [v,s]] = [[\lambda,v],s] + [v,[\lambda,s]]~,
\end{equation*}
for $\lambda \in \fh_0$, $v\in V$ and $s \in S$, is equivalent to the
$\lambda$-equivariance of the $[v,s]$ bracket and it is indeed
satisfied: this bracket is not zero but depends on $\varphi$ which is
left invariant by $\lambda$.  To go through all the identities
systematically, we use the notation $[ijk]$ for $i,j,k=0,1,2$ to
denote the identity involving $X \in \fh_{-i}$, $Y \in \fh_{-j}$ and
$Z \in \fh_{-k}$:
\begin{itemize}
\item the $[000]$ Jacobi identity is satisfied by virtue of
  $\fh_0=\fh_\varphi$ being a Lie algebra;
\item the $[001]$ and $[002]$ Jacobi identities are satisfied because
  $S$ and $V$ are $\fh_\varphi$-modules (by restriction);
\item the $[011]$, $[012]$ and $[022]$ Jacobi identities are
  satisfied because the $[SS]$, $[SV]$ Lie brackets are
  $\fh_\varphi$-equivariant;
\item the $[112]$ and $[111]$ Jacobi identities are satisfied by
  virtue of the first and second cocycle conditions \eqref{eq:cc1} and
  \eqref{eq:cc2}, respectively;
\item the $[122]$ and $[222]$ Jacobi identities are trivially
  satisfied to first order in $t$.
\end{itemize}

\subsection{All-orders integrability of the deformation}
\label{sec:all-order-integr}

Although the $[122]$ Jacobi identity is satisfied to first order in
$t$, it experiences an obstruction at second order.  Indeed, for all
$s\in S$ and $v_1,v_2 \in V$, the $[122]$ Jacobi identity is
\begin{equation*}
  [[v_1,v_2],s] \stackrel{?}{=} [v_1,[v_2,s]] - [v_2,[v_1,s]] = t^2
 [\beta^\varphi_{v_1},\beta^\varphi_{v_2}](s)~.
\end{equation*}
One can check that $\beta^\varphi_{v_1}\beta^\varphi_{v_2} \neq
\beta^\varphi_{v_2}\beta^\varphi_{v_1}$ in general, so that we need to cancel this
by modifying the $[v_1,v_2]$ bracket.  The following lemma suggests
how to do this.

\begin{lemma}
  \label{lem:betabeta}
  For all $v,w\in V$, the commutator $[\beta^\varphi_v,\beta^\varphi_w]$ lies in the image of
  $\fh_\varphi$ in $\End(S)$.
\end{lemma}

\begin{proof}
  There are three cases to consider, depending on whether $\varphi$ is
  timelike, spacelike or lightlike.  In all cases,
  $\varphi^2=\varphi\cdot\varphi\in\End(S)$ is a scalar multiple of
  the identity: positive if $\varphi$ is spacelike, negative if
  $\varphi$ is timelike and zero if $\varphi$ is lightlike.  In the
  first two cases, the $4$-plane $\Pi \subset V$ determined by
  $\varphi$ is nondegenerate and we may decompose
  $V = \Pi \oplus \Pi^\perp$.  We tackle these cases first and then
  finally the case where $\Pi$ is degenerate.

  \emph{$\Pi$ is nondegenerate.}  In this case $\varphi^2$ is a nonzero
  multiple of the identity.  If $v \in \Pi$, then
  $v \cdot \varphi = -\varphi \cdot v$ and $\beta^\varphi_v = 4 v \cdot \varphi$, whereas if
  $v \in \Pi^\perp$, then $v \cdot \varphi = \varphi \cdot v$ and $\beta^\varphi_v = -2
  v\cdot \varphi$.  In general, we can decompose any $v \in V$ as $v =
  v_\top + v_\perp$ with $v_\top \in \Pi$ and $v_\perp \in \Pi^\perp$,
  and $\beta^\varphi_v = 4 v_\top \cdot \varphi - 2 v_\perp \cdot \varphi$.  The
  commutator is given by
  \begin{equation*}
    \begin{split}
      [\beta^\varphi_v, \beta^\varphi_w] &= [4 v_\top \cdot \varphi - 2 v_\perp \cdot \varphi, 4
      w_\top \cdot \varphi - 2 w_\perp \cdot \varphi]\\
      &= 16 [v_\top \cdot \varphi, w_\top \cdot \varphi] + 4 [v_\perp \cdot \varphi,
      w_\perp \cdot \varphi] - 8 [v_\top \cdot \varphi, w_\perp \cdot \varphi] - 8
      [v_\perp \cdot \varphi, w_\top \cdot \varphi]\\
      &= - 16 [v_\top,w_\top] \cdot \varphi^2 + 4 [v_\perp,w_\perp] \cdot
      \varphi^2 +16 \eta(v_\top, w_\perp) \varphi^2 -16 \eta(v_\perp, w_\top) \varphi^2\\
      &= - 16 [v_\top,w_\top] \cdot \varphi^2 + 4 [v_\perp,w_\perp] \cdot
      \varphi^2~.
    \end{split}
  \end{equation*}
  Notice that $\varphi^2$ is a (nonzero) scalar endomorphism, so that
  $[\beta^\varphi_v,\beta^\varphi_w]$ lies in the image of $\fso(V)$ in $\Cl(V)$.
  Moreover, both $ [v_\top,w_\top]$ and $[v_\perp,w_\perp]$ commute
  with $\varphi$ in $\Cl(V)$, whence $[\beta^\varphi_v,\beta^\varphi_w]$ also lies in the
  image of $\stab(\varphi)$.

  \emph{$\Pi$ is degenerate.}  In this case, $V =
  \RR\left<\be_+,\be_-\right> \oplus \pi \oplus \pi^\perp$ and thus
  any $v \in V$ admits a unique decomposition $v = v_+ + v_- + v_\top
  + v_\perp$, where $v_\pm \in \RR\be_\pm$, $v_\top \in \pi$ and
  $v_\perp \in \pi^\perp$.  Now we still have that $v_\top \cdot
  \varphi = - \varphi \cdot v_\top$, $v_\perp \cdot \varphi = \varphi
  \cdot v_\perp$, but also $v_+ \cdot \varphi = \varphi \cdot v_+ = 0$
  and $v_- \cdot \varphi - \varphi \cdot v_- = - 2 \eta(v_-,\be_+) f$.
  Let us abbreviate $\eta(v,\be_+)$ by $\theta(v_-)$, so that $\varphi
  \cdot v_- \cdot \varphi = -2\theta(v_-) \Gamma_+ \cdot f^2$ and
  notice that $\beta^\varphi_v = 4  v_\top \cdot \varphi - 2 v_\perp
  \cdot \varphi - 2 v_- \cdot \varphi - 6 \theta(v_-) f$.  We now
  calculate (omitting the $\cdot$ notation):
  \begin{equation*}
    \begin{split}
      \tfrac14 [\beta^\varphi_v,\beta^\varphi_w] &= [2 v_\top \varphi - v_\perp 
      \varphi - v_- \varphi - 3\theta(v_-) f, 2 w_\top \varphi - w_\perp 
      \varphi - w_- \varphi - 3\theta(w_-) f]\\
      &= 4 [v_\top \varphi, w_\top \varphi] - 2 [v_\top \varphi, w_\perp \varphi] - 2 [v_\top
      \varphi, w_- \varphi] - 6 \theta(w_-) [v_\top \varphi, f] \\
      & \quad {} -2 [v_\perp \varphi, w_\top \varphi] + [v_\perp \varphi, w_\perp \varphi] +
      [v_\perp \varphi, w_- \varphi] + 3\theta(w_-) [v_\perp \varphi, f]\\
      & \quad {} -2 [v_- \varphi, w_\top \varphi] + [v_- \varphi, w_\perp \varphi] +
      [v_- \varphi, w_- \varphi] + 3\theta(w_-) [v_- \varphi, f]\\
      & \quad {} -6 \theta(v_-) [f,w_\top \varphi] + 3 \theta(v_-)
      [f,w_\perp \varphi] + 3 \theta(v_-) [f, w_- \varphi]~.
    \end{split}
  \end{equation*}
  Many of these terms vanish, namely:
  \begin{equation*}
    [v_\top \varphi, w_\top \varphi] = [v_\perp \varphi, w_\perp \varphi] = [v_\top \varphi, w_\perp
    \varphi]  = 0~, \qquad [v_\perp \varphi, f] = [v_- \varphi, f] = 0~,
  \end{equation*}
  whereas we have that
  \begin{equation*}
    \begin{aligned}[m]
      [v_\top \varphi, w_- \varphi] &= - 2 \theta(w_-) v_\top \Gamma_+ f^2\\
      [v_\top \varphi, f ] &= 2 v_\top \Gamma_+ f^2\\
      [v_\perp \varphi, w_- \varphi] &= -2 \theta(w_-) v_\perp \Gamma_+ f^2\\
      [v_- \varphi, w_- \varphi] &= 2 (\theta(v_-) w_- - \theta(w_-) v_-) \Gamma_+ f^2~.
    \end{aligned}
  \end{equation*}
  Putting it all together we arrive at
  \begin{equation*}
    \begin{split}
      \tfrac14 [\beta^\varphi_v,\beta^\varphi_w] &= \theta(v_-) (8
      w_\top + 2 w_\perp) \Gamma_+ f^2 - \theta(w_-) (8 v_\top + 2
      v_\perp) \Gamma_+ f^2 + (\theta(v_-) w_- - \theta(w_-) v_-)
      \Gamma_+ f^2\\
      &=\theta(v_-) (8 w_\top + 2 w_\perp) \Gamma_+ f^2 - \theta(w_-)
      (8 v_\top + 2 v_\perp) \Gamma_+ f^2~,
    \end{split}
  \end{equation*}
  where the last equality follows from the fact that both $v_-$ and
  $w_-$ are proportional to the vector $e_-$.  Since $f^2$ is a
  nonzero scalar multiple of the identity, we see that both terms in
  the RHS are in the image of $\fso(V)$ in $\Cl(V)$ and clearly also
  in the image of $\stab(\varphi)$, due to the presence of the
  $\Gamma_+$.
\end{proof}

Let us then define $\delta: \Lambda^2 V \to \fh_\varphi$ by
\begin{equation}
  \label{eq:delta}
  [\delta(v_1,v_2),s] = [\beta^\varphi_{v_1},\beta^\varphi_{v_2}](s)~,
\end{equation}
for all $s \in S$ and modify the $[VV]$ Lie bracket as
\begin{equation*}
  [v_1,v_2] = t^2 \delta(v_1,v_2)
\end{equation*}
so that the $[122]$ Jacobi identity is now satisfied to order $t^2$.
More is true, however, and all Jacobi identities are now satisfied for
all $t$.

We may summarise our results as follows:

\begin{theorem}
\label{thm:final}
  Let $\varphi \in \Lambda^4 V$ be decomposable and let
  $F = F_{\overline 0} \oplus F_{\overline 1}$ be a $\mathbb Z_2$-graded
  vector space with $F_{\overline 0} =V\oplus \fh_\varphi$ and
  $F_{\overline 1} = S$, where $\fh_\varphi = \fso(V)\cap\stab(\varphi)$.  The
  Lie brackets
  \begin{equation*}
    [v_1,v_2] = t^2 \delta(v_1,v_2)~, \qquad [v,s] = t \beta^\varphi_v(s)~, \qquad
    [s_1,s_2] = [s_1,s_2]_0 + t \gamma^\varphi(s_1,s_2)~,
  \end{equation*}
  with $\beta^\varphi$, $\gamma^\varphi$ and $\delta$ given by
  equations~\eqref{eq:beta} and \eqref{eq:delta}, together with the
  Dirac current $[s_1,s_2]_0$ as in \eqref{eq:DiracCurrent} and the
  adjoint action of $\fh_\varphi$ on itself and its actions on $S$ and
  $V$ given by restricting the spinor and vector representations of
  $\fso(V)$, respectively, define on $F$ a structure of a Lie
  superalgebra for all $t$.
	
  Moreover any filtered deformation of a $\mathbb Z$-graded subalgebra
  of the Poincaré superalgebra which differs only in zero degree is of
  this form.
\end{theorem}

\begin{proof}
  Two Jacobi components remain to be checked: the $[112]$ component,
  which is equivalent to
  \begin{equation}
    \label{eq:SSV}
    \delta(v,[s,s]_0)\stackrel{?}{=}
    2\gamma^\varphi(\beta^\varphi_v(s),s) \qquad \forall s\in S, v \in
    V~,
  \end{equation}
  and the $[222]$ component, which is equivalent to
  \begin{equation}
    \label{eq:VVV}
    [\delta(v,w),u] + [\delta(w,u),v] + [\delta(u,v),w]
    \stackrel{?}{=} 0 \qquad \forall u,v,w \in V~.
  \end{equation}

  As in the proof of Lemma~\ref{lem:betabeta}, we prove the identities
  \eqref{eq:SSV} and \eqref{eq:VVV} by calculating in $\Cl(V)$ and
  breaking the calculation into two cases, according to whether or not
  the plane associated with $\varphi$ is degenerate.  In proving
  identity~\eqref{eq:SSV} we will use the abbreviation $z := [s,s]_0$.

  \emph{$\Pi$ is nondegenerate.}  In the proof of
  Lemma~\ref{lem:betabeta} we derived the expression
  \begin{equation}
    \label{eq:deltaNonDeg}
    \delta(v,w) = -16 [v_\top, w_\top] \cdot \varphi^2 + 4 [v_\perp,
    w_\perp] \cdot \varphi^2~,
  \end{equation}
  for any $v,w \in V$. Let us calculate, for $z=[s,s]_0$,
  \begin{equation*}
    \eta([\delta(v,z),w],u) = -16
    \eta([[v_\top,z_\top],w_\top],u_\top)+ 4
    \eta([[v_\perp,z_\perp],w_\perp],u_\perp)~.
  \end{equation*}
  Using that
  \begin{equation}
    \label{eq:nestedcomm}
    [[v,z],w]= 4 \eta(v,w) z - 4 \eta(z,w) v~,
  \end{equation}
  we can write
  \begin{multline*}
    \eta([\delta(v,z),w],u) = - 64 \eta(v_\top, w_\top) \eta(z_\top,
    u_\top) + 64 \eta(z_\top, w_\top) \eta(v_\top, u_\top)\\
    + 16 \eta(v_\perp, w_\perp) \eta(z_\perp, u_\perp) - 16
    \eta(z_\perp, w_\perp) \eta(v_\perp, u_\perp)~.
  \end{multline*}
  On the other hand,
  \begin{equation*}
   \eta (2[\gamma^\varphi(\beta^\varphi_v(s),s),w],u) = -2 \left<u
      \cdot \beta^\varphi_w \cdot \beta^\varphi_v \cdot s, s\right> -
    2 \left<\sigma(\beta^\varphi_v) \cdot u \cdot \beta^\varphi_w
      \cdot s, s\right>~,
  \end{equation*}
  where $\sigma$ is the anti-involution in $\Cl(V)$ defined by the
  symplectic structure.  One calculates in $\Cl(V)$ to find
  \begin{equation*}
    \sigma(\beta^\varphi_v) \cdot u \cdot \beta^\varphi_w = 4 (2
    v_\top + v_\perp)(u_\top-u_\perp)(2w_\top+w_\perp)~,
  \end{equation*}
  and
  \begin{equation*}
    u \cdot \beta^\varphi_w \cdot \beta^\varphi_v  = 4 (u_\top +
    u_\perp) (w_\perp v_\perp + 2 w_\perp v_\top - 2 w_\top v_\perp -
    4 w_\top v_\top)~,
  \end{equation*}
  and hence
  \begin{equation*}
    \eta (2[\gamma^\varphi(\beta^\varphi_v(s),s),w],u) = -16 \left<s,
      [[u_\top, w_\top], v_\top] \cdot s\right>+ 4 \left<s,
      [[u_\perp, w_\perp], v_\perp] \cdot s\right>~.
  \end{equation*}
  We use equation~\eqref{eq:nestedcomm} again to arrive at
  \begin{multline*}
    \eta (2[\gamma^\varphi(\beta^\varphi_v(s),s),w],u) = -64
    \eta(u_\top, v_\top) \left<s, w_\top \cdot s\right> + 64
    \eta(w_\top, v_\top) \left<s, u_\top \cdot s\right>\\ + 16
    \eta(u_\perp, v_\perp) \left<s, w_\perp \cdot s\right> - 16
    \eta(w_\perp, v_\perp) \left<s, u_\perp \cdot s\right>~,
  \end{multline*}
  which agrees with $\eta([\delta(v,z),w],u)$ after using the
  definition of the Dirac current.

  To prove the identity~\eqref{eq:VVV}, we again depart from the
  expression~\eqref{eq:deltaNonDeg} for $\delta(v,w)$, so that in
  $\Cl(V)$,
  \begin{equation*}
    \begin{split}
      [\delta(v,w), u ] &= [-16 [v_\top, w_\top] \cdot \varphi^2 + 4
      [v_\perp, w_\perp] \cdot \varphi^2, u_\top + u_\perp]\\
      &= -16 [[v_\top, w_\top], u_\top] \cdot \varphi^2 + 4 [[v_\perp,
      w_\perp], u_\perp] \cdot \varphi^2~,
    \end{split}
  \end{equation*}
  using that $\varphi^2$ is central and the fact that
  $[v_\top,w_\top] \in \fso(\Pi)$ (resp. $[v_\perp,w_\perp] \in
  \fso(\Pi^\perp)$) acts trivially on $\Pi^\perp$ (resp. $\Pi$).  It
  is clear that the $[222]$ Jacobi identity follows in this case from
  the Jacobi identity of the commutator in the associative algebra
  $\Cl(V)$.
  
  \emph{$\Pi$ is degenerate.}  This case is computationally more
  involved, but it is again simply a calculation in $\Cl(V)$.  Let us
  prove first the identity~\eqref{eq:VVV}.  In the proof of
  Lemma~\ref{lem:betabeta} we showed that
  \begin{equation*}
    \delta(v,w) = 8 \theta(v_-) (4 w_\top + w_\perp) \cdot \varphi \cdot f -
    8 \theta(w_-) (4 v_\top + v_\perp) \cdot \varphi \cdot f~,
  \end{equation*}
  where we recall that $\theta(v_-) = \eta(v,\be_+)$.  Therefore in
  $\Cl(V)$,
  \begin{equation*}
    \begin{split}
      [\delta(v,w),u] &= [\delta(v,w), u_\top + u_\perp + u_-+ u_+]\\
      &= +8 \theta(v_-) [(4 w_\top + w_\perp) \cdot \varphi  \cdot f,
      u_\top + u_\perp + u_-]\\
      & \quad {} - 8 \theta(w_-) [(4 v_\top + v_\perp)  \cdot \varphi
      \cdot f, u_\top + u_\perp + u_-]
    \end{split}
  \end{equation*}
  where we have used that $\delta(v,w)$ leaves $\be_+$ invariant.
  Next we use the following results:
  \begin{equation*}
    \begin{aligned}[m]
      [w_\top  \cdot \varphi  \cdot f, u_\top] &= 2 \eta(w_\top, u_\top) \varphi \cdot  f\\
      [w_\perp  \cdot \varphi  \cdot f, u_\perp] &= 2 \eta(w_\perp, u_\perp) \varphi \cdot  f\\
      [w_\top \cdot  \varphi \cdot  f, u_\perp] &= 0
   \end{aligned}
    \qquad\qquad
    \begin{aligned}[m]   
      [w_\perp  \cdot \varphi  \cdot f, u_\top] &= 0\\
      [w_\top  \cdot \varphi  \cdot f, u_-] &= - 2 \theta(u_-) w_\top  \cdot f^2\\
      [w_\perp  \cdot \varphi  \cdot f, u_-] &= - 2 \theta(u_-) w_\perp \cdot  f^2~,
    \end{aligned}
  \end{equation*}
  and arrive at
  \begin{equation*}
    \begin{split}
    \tfrac18 [\delta(v,w),u] &= 8 \theta(v_-) \eta(w_\top,u_\top) \varphi \cdot f - 8
    \theta(w_-) \eta(v_\top, u_\top) \varphi \cdot f \\
    & \quad {} + 2 \theta(v_-) \eta(w_\perp,u_\perp) \varphi \cdot f - 2 \theta(w_-)
    \eta(v_\perp,u_\perp) \varphi \cdot f \\
    & \quad {} - 8 \theta(u_-) \theta(v_-) w_\top  \cdot f^2 + 8 \theta(u_-)
    \theta(w_-) v_\top  \cdot f^2\\
    & \quad {} - 2 \theta(u_-) \theta(v_-) w_\perp  \cdot f^2 + 2 \theta(u_-)
    \theta(w_-) v_\perp  \cdot f^2~,
    \end{split}
  \end{equation*}
  which when inserted in the Jacobi identity vanishes, due to the
  terms cancelling pairwise, thus proving the identity~\eqref{eq:VVV}.
  The identity~\eqref{eq:SSV} is proved in a similar way, so we will
  be brief.  We now find that the left-hand side of \eqref{eq:SSV} is
  given by
  \begin{equation*}
    -8 \theta(v_-) (4 z_\top + z_\perp) \Gamma_+  + 8 \theta(z_-) (4
    v_\top + v_\perp) \Gamma_+~,
  \end{equation*}
  and this is precisely what we obtain for the right-hand side.

  Finally, the last claim of the theorem follows from
  Theorem~\ref{thm:first}.
\end{proof}

In summary, we find three isomorphism classes of nontrivial filtered
deformations of $\ZZ$-graded subalgebras of the Poincaré superalgebra
$\fg$ which differ only in degree zero.  They are characterised by a
decomposable $\varphi \in \Lambda^4 V$.  Such a $\varphi$ defines a
stabiliser $\fh_\varphi \subset \fso(V)$ and also a filtered
deformation of the $\ZZ$-graded subalgebra
$\fh_\varphi \oplus S \oplus V \subset \fg$ given by
\begin{equation}
  \label{eq:FLSA}
  \begin{aligned}[m]
    [A,B] &= AB - BA\\
    [A,s] &= As\\
    [A,v] &= Av
  \end{aligned}
  \qquad\qquad
  \begin{aligned}[m]
    [s,s] &= [s,s]_0 + t \gamma^\varphi(s,s)\\
    [v,s] &= t \beta^\varphi(v,s)\\
    [v,w] &= t^2 \delta(v,w)~,
  \end{aligned}
\end{equation}
for all $A,B \in \fh_\varphi$, $s \in S$ and $v,w \in V$, and where
the maps $\beta^\varphi : V \otimes S \to S$ and $\gamma^\varphi :
\odot^2 S \to \fh_\varphi$ are as in equation~\eqref{eq:beta} and $\delta:
\Lambda^2 V \to \fh_\varphi$ is as in equation~\eqref{eq:delta}.

By Lemma~\ref{lem:varphi}, $\CSO(V)$-related $\varphi$'s give rise
to isomorphic filtered deformations, so it is enough to choose a
representative $\varphi$ from each orbit.  A possible choice is the
following:
\begin{enumerate}
\item $\varphi = \be_0 \wedge \be_1 \wedge \be_2 \wedge \be_3$, where
  $(\be_\mu)$ is an $\eta$-orthonormal basis for $V$.  The stabiliser
  is $\fh_\varphi \cong \fso(1,3) \oplus \fso(7)$.  The Lie brackets
  on $\fh_\varphi \oplus V$ can be read from
  Lemma~\ref{lem:betabeta}, and we find that they give rise to a Lie
  algebra isomorphic to $\fso(2,3) \oplus \fso(8)$.  This is the Lie
  algebra of isometries of the Freund--Rubin backgrounds $\AdS_4 \times
  S^7$.  The resulting Lie superalgebra on $\fh_\varphi \oplus S
  \oplus V$ is isomorphic to the Killing superalgebra of this family
  of backgrounds; namely, $\fosp(8|4)$.
\item $\varphi = \be_7 \wedge \be_8 \wedge \be_9 \wedge  \be_\ten$,
  where again $(\be_\mu)$ is an $\eta$-orthonormal basis for $V$.  The
  stabiliser is $\fh_\varphi \cong \fso(4) \oplus \fso(1,6)$.  The Lie
  brackets on $\fh_\varphi \oplus V$ are isomorphic to $\fso(5) \oplus
  \fso(2,6)$, which is the isometry Lie algebra of the Freund--Rubin
  backgrounds $S^4 \times \AdS_7$.  The resulting filtered deformation
  is isomorphic to the Lie superalgebra $\fosp(2,6|4)$, which is the
  Killing superalgebra of this family of backgrounds.
\item $\varphi = \be_+ \wedge \be_1 \wedge \be_2 \wedge \be_3$,
  relative to an $\eta$-Witt basis $(\be_+,\be_-,\be_i)$ for $V$.  The
  stabiliser is $\fh_\varphi = (\fso(3) \oplus \fso(6)) \ltimes \RR^9$
  and the Lie brackets on $\fh_\varphi \oplus V$ give it the structure
  of a Lie algebra isomorphic to the isometry Lie algebra of the
  Cahen--Wallach spacetime underlying the Kowalski-Glikman pp-wave.
  The resulting Lie superalgebra is isomorphic to the Killing
  superalgebra of the Kowalski-Glikman wave, which is itself a
  contraction (in the sense of Inönü--Wigner) of both of the
  Freund--Rubin Killing superalgebras.
\end{enumerate}

In summary, we recover the classification of maximally supersymmetric
vacua of 11-di\-men\-sional supergravity via their Killing
superalgebras.

\section{Discussion and conclusion}
\label{sec:disc-concl}

In this paper we have determined the (isomorphism classes of) Lie
superalgebras which are filtered deformations of $\mathbb Z$-graded
subalgebras $\fh = V\oplus S\oplus \fh_0$, with
$\fh_0 \subset \fso(V)$, of the eleven-di\-men\-sional Poincaré
superalgebra. We have found that aside from the Poincaré superalgebra
itself ($\fh_0 = \fso(V)$) and its $\mathbb Z$-graded subalgebras,
there are three other Lie superalgebras corresponding to the symmetry
superalgebras of the non-flat maximally supersymmetric backgrounds of
eleven-di\-men\-sional supergravity: the two (families of) Freund--Rubin
backgrounds and their common Penrose limit.

In so doing we have recovered by cohomological means the connection
$D$ on the spinor bundle which is defined by the supersymmetry
variation of the gravitino.  We could say that we have, in a very real
sense, rediscovered eleven-di\-men\-sional supergravity from the Spencer
cohomology of the Poincaré superalgebra.

More remarkable still is perhaps the fact that the classification of
nontrivial filtered deformations of subalgebras of the Poincaré
superalgebra precisely agrees with the classification of Killing
superalgebras of non-flat maximally supersymmetric backgrounds of
eleven-di\-men\-sional supergravity.  To be clear, what is remarkable is
not that we recover these Killing superalgebras -- after all, it can
be shown in full generality that the symmetry superalgebra of a
supersymmetric background is a filtered deformation of some subalgebra
of the Poincaré superalgebra --- but that we find \emph{no other}
filtered deformations.  We interpret this as encouraging evidence as
to the usefulness of both the notions of super Poincaré structures and
of symmetry superalgebras as organisational tools in the
classification problem of supersymmetric supergravity
backgrounds.

An interesting question is whether every filtered deformation of a
subalgebra of the Poincaré superalgebra is geometrically realised as
the Killing superalgebra of a supersymmetric background.  First of
all, as shown by the (undeformed) subalgebras of the Poincaré
superalgebra, these are only contained in the maximal such
superalgebra (namely, the Poincaré superalgebra itself).  This is not
surprising since it is only the supertranslation ideal which is
actually generated by the Killing spinors.  More worrying, though, are
examples of filtered deformations which are not yet known to be
realised geometrically (such as the deformation of the M2-brane
Killing superalgebra found in \cite{JMFSuperDeform}, which suggests
very strongly the existence of a half-BPS black anti-de~Sitter
membrane, whose construction continues to elude us), or those such as
the putative $N=28$ pp-wave conjectured in \cite{Fernando:2004jt} and
which was shown in \cite{Gran:2010tj} not to exist.

Before concluding, we would also like to mention an interesting
relation with the off-shell pure spinor superfield formulation of
eleven-dimensional supergravity (see, e.g., the review
\cite{Cederwall:2013vba} and references therein). The starting point
of the pure spinor approach is the observation that the bosonic
equations of motion of eleven-di\-men\-sional supergravity reside in
the direct sum of irreducible $\fso(V)$-modules with Dynkin labels
$[11000]$ and, respectively, $[10002]$,
cf.~\cite[equation~(4.14)]{Cederwall:2013vba}. Now pure spinors are
the Dirac spinors $s\in S \otimes \CC$ with vanishing Dirac current 
$k(s,s)$. For them, the associated supercharge $Q$ satisfies $Q^2=0$
and one can see that the cohomology of $Q$ encodes the (linearised)
equations of motion.

In our approach one can check that the Spencer cohomology group
$H^{0,2}(\fm,\fg)$ of the Poincaré superalgebra $\fg$ is isomorphic
to $[11000]\oplus[10002]$, i.e., it encodes the equations of motion.
This fact suggests the possibility of modifying the definition of a
super Poincaré structure $(M,\mathcal D)$ as a Tanaka structure
whose ``symbol space'' $\fm(x)$ has been deformed along directions in
$H^{0,2}(\fm,\fg)$. It might be interesting to investigate these more
general Tanaka structures and understand differences and similarities
with the pure spinor approach, also in view of possible applications
to the construction of off-shell formulations of supergravity
theories.

\section*{Acknowledgments}

The research of JMF is supported in part by two grants from the UK
Science and Technology Facilities Council (ST/J000329/1 and
ST/L000458/1), and that of AS is fully supported by a Marie-Curie
research fellowship of the ''Istituto Nazionale di Alta Matematica''
(Italy). We are grateful to our respective funding agencies for their
support. During the final stretch of writing, JMF was a guest of the
Grupo de Investigación ``Geometría diferencial y convexa'' of the
Universidad de Murcia, and he wishes to thank Ángel Ferrández
Izquierdo for the invitation, the hospitality and for providing such a
pleasant working atmosphere.

\appendix

\section{Clifford conventions}
\label{sec:clifford-conventions}

The proofs of a couple of results are easier if we work relative to a
basis for the Clifford algebra.  In this appendix we set out the
conventions which will be employed in this paper, especially to prove
Proposition~\ref{prop:beta} and Lemma~\ref{lem:quadrics}.

We start with some properties of the Clifford algebra associated to an
eleven-dimensional Lorentzian vector space $(V,\eta)$ with ``mostly
minus'' signature.  The Clifford algebra $\Cl(V)$, with relations
\begin{equation*}
  v^2 = - \eta(v,v) \1 \qquad\forall v \in V~,
\end{equation*}
is isomorphic as a real associative algebra to two copies of the
algebra of real $32 \times 32$ matrices.  It follows from this
isomorphism that $\Cl(V)$ has two inequivalent irreducible Clifford
modules, which are real and of dimension $32$.

The Clifford algebra $\Cl(V)$ is filtered (and $\ZZ_2$-graded) and the
associated graded algebra is the exterior algebra $\Lambda^\bullet
V$.  An explicit vector space isomorphism $\Lambda^\bullet V
\stackrel{\cong}{\rightarrow} \Cl(V)$ can be described as follows.

Let $(\be_\mu)$, for $\mu=0,1,\dots,9,\ten$, be an $\eta$-orthonormal
basis; that is,
\begin{equation*}
  \eta(\be_\mu,\be_\nu) = \eta_{\mu\nu} =
  \begin{pmatrix}
    1 & 0 \\ 0 & -I_{10}
  \end{pmatrix}~.
\end{equation*}
The Clifford algebra $\Cl(V)$ is generated by the image of $V$
under the map $V \to \Cl(V)$ which sends $\be_\mu$ to
$\Gamma_\mu$, with
\begin{equation*}
  \Gamma_\mu \Gamma_\nu + \Gamma_\nu \Gamma_\mu = - 2 \eta_{\mu\nu} \1~.
\end{equation*}
Notice that due to our choice of a mostly minus $\eta$, $\Gamma_0^2 =
-\1$.  We use the notation $\Gamma_{\mu_1\cdots\mu_p}$ for the totally
antisymmetric product
\begin{equation*}
  \Gamma_{\mu_1\cdots\mu_p} = \Gamma_{[\mu_1} \Gamma_{\mu_2} \cdots
  \Gamma_{\mu_p]} := \tfrac1{p!} \sum_{\sigma \in S_p} (-1)^\sigma
  \Gamma_{\mu_{\sigma(1)}} \cdots \Gamma_{\mu_{\sigma(p)}}~,
\end{equation*}
with $S_p$ the symmetric group in $\{1,\dots,p\}$ and $(-1)^\sigma$
the sign of the permutation $\sigma \in S_p$.

The explicit isomorphism $\Lambda^\bullet V \to \Cl(V)$ is built out
of the maps $\Lambda^p V \to \Cl(V)$ given by sending
\begin{equation*}
  \be_{\mu_1} \wedge \dots \wedge \be_{\mu_p} \mapsto
  \Gamma_{\mu_1\dots\mu_p}
\end{equation*}
and extending linearly.  Thus an $\eta$-orthonormal basis for $V$
induces a basis for $\Cl(V)$ given by the $\Gamma_{\mu_1\cdots\mu_p}$
for $p=0,1,\dots,11$.  The volume element
$\Gamma_{11} = \Gamma_0 \Gamma_1 \cdots \Gamma_\ten$ is central in
$\Cl(V)$ and satisfies $\Gamma_{11} \Gamma_{11} = \1$.  The two
non-isomorphic irreducible Clifford modules $S_\pm$ of $\Cl(V)$ are
distinguished by the action of $\Gamma_{11}$, where $\Gamma_{11}$ acts
like $\pm \1$ on $S_\pm$.  We will work with $S = S_-$ in this paper.

Endomorphisms of $S$ can be described in terms of elements of
$\Cl(V)$.  A basis for the endomorphisms of $S$ by the image in
$\End(S)$ of $\Gamma_{\mu_1\cdots\mu_p}$ for $p=0,1,\dots,5$.  We will
often tacitly use this isomorphism
\begin{equation*}
  \End(S)\cong\bigoplus_{p=0}^5\Lambda^p V
\end{equation*}
in the paper and let $p$-forms with $p=0,\ldots,5$ act on $S$.

There is an action of $\fso(V)$ on $S$ via the embedding of $\fso(V)$ in
$\Cl(V)$.  This is described as follows.  If
$L_{\mu\nu} = - L_{\nu\mu} \in \fso(V)$ is defined by
\begin{equation*}
  L_{\mu\nu} (\be_\rho) = \eta_{\rho\nu} \be_\mu - \eta_{\rho\mu} \be_\nu~
\end{equation*}
then it is embedded in $\Cl(V)$ as
\begin{equation}
  \label{eq:spingens}
  L_{\mu\nu} \mapsto -\tfrac12 \Gamma_{\mu\nu} = -\tfrac14
  [\Gamma_\mu, \Gamma_\nu]~.
\end{equation}
Indeed, we have the following commutator in $\Cl(V)$:
\begin{equation*}
  \left[-\tfrac12 \Gamma_{\mu\nu}, \Gamma_\rho\right] = \eta_{\rho\nu}
  \Gamma_\mu - \eta_{\rho\mu} \Gamma_\nu~.
\end{equation*}

On $S$ we have an $\fso(V)$-equivariant symplectic structure
$\left<-,-\right>$.  Relative to a Majorana basis for $S$ where the
$\Gamma_\mu$ are represented by \emph{real} matrices, we can choose
the symplectic structure defined by the matrix representing
$\Gamma_0$.  If $s_1,s_2 \in S$, it is often convenient to write
$\left<s_1,s_2\right>$ as $\overline{s_1}s_2$.  We have that
\begin{equation*}
  \overline{s_1} \Gamma_{\mu_1\cdots\mu_p} s_2 = \varepsilon_p
  \overline{s_2} \Gamma_{\mu_1\cdots\mu_p} s_1~,
\end{equation*}
where $\varepsilon_p=+1$ for $p=1,2,5$ and $\varepsilon_p = -1$ for
$p=0,3,4$,  which reflects the isomorphisms of $\fso(V)$-modules
\begin{equation*}
  \Lambda^2 S \cong \Lambda^0V \oplus \Lambda^3V \oplus \Lambda^4V
  \qquad\text{and}\qquad \odot^2 S \cong V \oplus \Lambda^2V \oplus
  \Lambda^5V~.
\end{equation*}
Three easy consequences of this fact are the following:
\begin{enumerate}
\item for $v \in V$,
  \begin{equation}
    \label{eq:nondeg1forms}
    \left<s, v \cdot s\right> = 0\qquad \forall s \in S \implies
    v = 0~;
  \end{equation}
\item for $\zeta \in \Lambda^2 V$,
  \begin{equation}
    \label{eq:nondeg2forms}
    \left<s, \zeta \cdot s\right> = 0\qquad \forall s \in S \implies
    \zeta = 0~;
  \end{equation}
\item and for $\theta \in \Lambda^5 V$,
  \begin{equation}
    \label{eq:nondeg5forms}
    \left<s, \theta \cdot s\right> = 0\qquad \forall s \in S \implies
    \theta = 0~.
  \end{equation}
\end{enumerate}
Another consequence of this fact is the following isomorphism of
$\fso(V)$-modules
\begin{equation*}
  S\otimes S \cong \bigoplus_{p=0}^5\Lambda^p V~,
\end{equation*}
where the $\fso(V)$-submodule of $S \otimes S$ isomorphic to
$\Lambda^q V$ is given by
\begin{equation}
  \label{eq:lambdaVinSS}
  \Lambda^q V\cong\left\{\sum s_i\otimes
    s^{'}_i~ \middle | ~ \sum\overline{s_i} \Gamma_{\mu_1\cdots\mu_p}
    s'_i=0~\text{for all}~p\neq q,\, 0\leq p\leq 5\right\}
\end{equation}
for all $q=0,\ldots,5$ and that $\sbar\Gamma_{\mu_1\cdots\mu_p} s = 0$
except when $p=1,2,5$.

On occasion we will also need to use an $\eta$-Witt basis
$(\be_+,\be_-,\be_i)$, with $i=1,\dots,9$, for $V$, where
$\eta(\be_+,\be_-) =1$ and $\eta(\be_i,\be_j) = -\delta_{ij}$.  Given
an $\eta$-orthonormal basis, we may obtain an $\eta$-Witt basis by
$\be_{\pm} = \frac1{\sqrt2} (\be_0 \pm \be_\ten)$ and
$\be_1,\dots,\be_9$ coinciding.  The image in $\Cl(V)$ of $\be_\pm$
will be denoted $\Gamma_\pm$ and obey $\left(\Gamma_\pm\right)^2 = 0$.

Finally, we record a number of useful identities to perform
calculations in the Clifford algebra.

If $v \in V$ and $\theta \in \Lambda^p V$ their Clifford product in
$\Cl(V)$ satisfy
\begin{equation}
  \label{eq:VonForms}
  \begin{split}
    v \cdot \theta &= v \wedge \theta - \iota_v \theta\\
    \theta \cdot v &= (-1)^p \left( v \wedge \theta + \iota_v \theta\right)~.
  \end{split}
\end{equation}

The Fierz identity expresses the rank-one endomorphism
$s_1\overline {s_2}$ of $S$ defined by $(s_1\overline{s_2})(s) =
\left<s_2,s\right> s_1$ in terms of the standard basis of $\End(S)$.
We shall only need the special case where $s_1 = s_2$.  The identity
reads
\begin{equation}
  \label{eq:Fierz}
  s\sbar = -\tfrac1{32} \left(\sbar \Gamma^\mu s
    \Gamma_\mu + \tfrac12 \sbar \Gamma^{\mu\nu} s
    \Gamma_{\mu\nu} + \tfrac1{5!} \sbar \Gamma^{\mu_1\cdots\mu_5} s
    \Gamma_{\mu_1\cdots\mu_5} \right)~.
\end{equation}

The following identities come in handy when using the Fierz identity:
\begin{equation*}
  \begin{split}
    \Gamma_\mu \Gamma_{\nu_1\cdots\nu_p} \Gamma^\mu &=
    (-1)^{p+1}(11-2p) \Gamma_{\nu_1\cdots\nu_p}\\
    \Gamma_{\mu_1\mu_2} \Gamma_{\nu_1\cdots\nu_p} \Gamma^{\mu_1\mu_2} &=
    (11-(11-2p)^2) \Gamma_{\nu_1\cdots\nu_p}\\
    \Gamma_{\mu_1\cdots\mu_5} \Gamma_{\nu_1\cdots\nu_p}
    \Gamma^{\mu_1\cdots\mu_5} &= (-1)^{p+1} \left((11-2 p)^4-90 (11-2
      p)^2+1289\right) (11-2 p) \Gamma_{\nu_1\cdots\nu_p}~,
  \end{split}
\end{equation*}
where $\Gamma^\mu$ is defined by
$\Gamma_\nu=\eta_{\nu\mu}\Gamma^\mu$.  It follows from these
identities that
\begin{equation}
  \label{eq:trilinears}
  \tfrac12 \sbar \Gamma_{\mu\nu} s \Gamma^{\mu\nu} s = 5 \sbar
  \Gamma_\mu s \Gamma^\mu s \qquad\text{and}\qquad
  \tfrac1{5!} \sbar \Gamma_{\mu_1\cdots \mu_5} s \Gamma^{\mu_1\cdots
    \mu_5} s = -6 \sbar \Gamma_\mu s \Gamma^\mu s~,
\end{equation}
consistent with the fact that the endomorphism $s\sbar$ annihilates
$s$ due to the symplectic nature of the spinor inner product.

\section{Some representations of \texorpdfstring{$\fso(V)$}{so(V)}}
\label{sec:some-representations}

The Lie algebra $\fso(V)$ is a real form of the complex simple Lie
algebra of type $B_5$. We will therefore use the Dynkin label
$[n_1\dots n_5]$, $n_i \in \NN$, to refer to the (real) irreducible
module with highest weight $\sum_i n_i \lambda_i$, where $\lambda_i$
are a choice of fundamental weights. The following dictionary is
helpful. The module $V$ has Dynkin label $[10000]$, whereas the
adjoint module $\fso(V) \cong \Lambda^2 V$ has label $[01000]$ and the
spinor module $S$ has label $[00001]$.  Other representations which
will play a rôle are shown in Table~\ref{tab:some-irred-modul}.  The
representations with a 0 subscript are the kernels of Clifford
multiplication inside $V \otimes S$ or $V\otimes\Lambda^p V$ with
$p\geq 1$. In other words, they are the irreducible $\fso(V)$-modules
defined by the short exact sequences:
\begin{gather*}
  0 \rightarrow (V \otimes S)_0 \rightarrow V \otimes
  S \stackrel{\text{cl}}{\rightarrow} S \rightarrow 0\\
  0 \rightarrow (V \otimes \Lambda^pV)_0 \rightarrow V \otimes
  \Lambda^p V \stackrel{\text{cl}}{\rightarrow} \Lambda^{p-1}V \oplus
  \Lambda^{p+1}V \rightarrow 0~,
\end{gather*}
where $\Lambda^1V= V$ and $\Lambda^0V = \RR$.  Notice that for $p=1$,
there is an isomorphism of modules $(V\otimes\Lambda^1 V)_0 \cong
\odot^2_0 V$, the $\eta$-traceless symmetric square of $V$.

\begin{table}[h!]
  \centering
  \caption{Some irreducible modules of $\fso(V)$}
  \label{tab:some-irred-modul}
  \begin{tabular}{c*{2}{|>{$}r<{$}}}
    \multicolumn{1}{c|}{Label} & \multicolumn{1}{c|}{Alias} & \multicolumn{1}{c}{$\dim$}\\\hline
    {} [00000] & \RR & 1\\
    {} [10000] & V & 11\\
    {} [00001] & S & 32\\
    {} [01000] & \Lambda^2 V & 55\\
    {} [00100] & \Lambda^3 V & 165\\
    {} [00010] & \Lambda^4 V & 330\\
    {} [00002] & \Lambda^5 V & 462
  \end{tabular}
  \qquad\qquad
  \begin{tabular}{c*{2}{|>{$}r<{$}}}
    \multicolumn{1}{c|}{Label} & \multicolumn{1}{c|}{Alias} & \multicolumn{1}{c}{$\dim$}\\\hline
    {} [10001] & (V \otimes S)_0 & 320\\
    {} [20000] & \odot^2_0V & 65\\
    {} [11000] & (V \otimes \Lambda^2V)_0 & 429\\
    {} [10100] & (V \otimes \Lambda^3V)_0 & 1430\\
    {} [10010] & (V \otimes \Lambda^4V)_0 & 3003\\
    {} [10002] & (V \otimes \Lambda^5V)_0 & 4290
  \end{tabular}
\end{table}

\bibliographystyle{utphys}
\bibliography{Spencer}

\end{document}